\documentclass[subscriptcorrection,upint,varvw,barcolor=Goldenrod3,mathalfa=cal=euler,balance,hyphenate,french,nolists]{asmejour} %

\hypersetup{%
	pdfauthor={John H. Lienhard},                       		   	
	pdftitle={ASME Journal Paper LaTeX Template},                  	
	pdfkeywords={ASME journal paper, LaTeX template, BibTeX style, asmejour class},
	pdfsubject = {Describes the asmejour LaTeX template},			
	pdflicenseurl={https://ctan.org/pkg/asmejour},
}

\usepackage{amsmath}
\usepackage{xcolor}
\usepackage{graphicx}
\usepackage{tikz}
\usetikzlibrary{arrows.meta,calc}
\usetikzlibrary{arrows.meta}
\newtheorem{theorem}{Theorem}
\newtheorem{definition}{Definition}

\newtheorem{proposition}{Proposition}
\definecolor{blockblue}{RGB}{0,0,0}
\usetikzlibrary{arrows, positioning, calc}
\usetikzlibrary{arrows, positioning, calc, arrows.meta}
\usepackage{environ}

\makeatletter
\newsavebox{\measure@tikzpicture}
\NewEnviron{scaletikzpicturetowidth}[1]{%
	\def\tikz@width{#1}%
	\def\tikzscale{1}\begin{lrbox}{\measure@tikzpicture}%
        		\leavevmode\BODY
	\end{lrbox}%
	\pgfmathparse{#1/\wd\measure@tikzpicture}%
	\edef\tikzscale{\pgfmathresult}%
	\leavevmode\BODY
    }
\JourName{Dyn. Sys., Meas., Control}

\begin{document}



\SetAuthorBlock{Ali Hosseini\CorrespondingAuthor}{%
Department of Precision and Microsystems Engineering,\\
Delft University of Technology, \\
2628 CD Delft, The Netherlands \\
email: s.a.hosseini@tudelft.nl
}

\SetAuthorBlock{Dragan Kosti\'c}{%
Center of Competency,\\
ASMPT,\\
6641 TL Beuningen, The Netherlands \\
email: dragan.kostic@asmpt.com
}

\SetAuthorBlock{Hassan HosseinNia}{%
Department of Precision and Microsystems Engineering,\\
Delft University of Technology, \\
2628 CD Delft, The Netherlands \\
email: s.h.hosseinniakani@tudelft.nl
}

\title{Reliability Assessment and Performance Enhancement of Reset Control Systems}


\begin{abstract}
This paper develops a frequency-domain reliability assessment framework for reset control systems. The closed-loop higher-order sinusoidal-input describing function formulation is extended to explicitly include the reset-triggering signal generated through a shaping filter. Based on this signal, two metrics are introduced: \(\sigma_t\), which quantifies reset-time deviation, and \(\sigma_d\), which evaluates the tendency toward additional zero crossings. These metrics provide design-oriented indicators for identifying potentially unreliable reset behavior. To improve reset-triggering reliability, a first-order shaping filter is proposed for a generalized first-order reset element, increasing the low-frequency attenuation slope of the nonzero higher-order harmonics. The proposed analysis is evaluated on an industrial motion stage. The results show that the proposed metrics capture reliability issues that are not evident from the first-order closed-loop response alone and can therefore support the design of reset controllers with more reliable reset-triggering behavior.

\end{abstract}

\date{Version \versionno, \today}

\maketitle 
\begingroup
\renewcommand\thefootnote{}
\footnotetext{This project is co-financed by ASMPT and Holland High Tech, Top Sector High Tech Systems and Materials, through a PPS innovation grant for public-private collaboration in research and development.}
\addtocounter{footnote}{-1}
\endgroup

\section{Introduction}
\label{sec: Introduction}

Reset elements~\cite{clegg1958nonlinear,beker2004fundamental,guo2009frequency,ExpDemonst,ResetIntServo,secondOrderMarcel} are nonlinear filters that have been used to relax some of the fundamental performance limitations of linear time-invariant (LTI) control systems~\cite{waterbed2012,freudenberg2000surveyBodeGain}. Since reset control systems are nonlinear, their frequency-domain analysis cannot be performed using conventional LTI sensitivity functions alone. A widely used tool for this purpose is the sinusoidal-input describing function (SIDF), which characterizes the first harmonic of the reset-element response to a sinusoidal input~\cite{guo2009frequency}. To account for the non-sinusoidal response generated by reset actions, higher-order SIDFs (HOSIDFs) were later introduced in~\cite{saikumar2021loop}. These tools enable an approximate frequency-domain description of the higher-order harmonic content generated by reset elements.

For closed-loop reset control systems, the relation between the open-loop HOSIDFs of the reset element and the corresponding closed-loop frequency-domain response was studied in~\cite{saikumar2021loop}. This framework was further generalized in~\cite{LukeLure} by representing the closed-loop reset control system in a Lur'e-type form~\cite{khalil2002nonlinear}. This representation allows a broad class of interconnections between LTI elements and a reset element to be analyzed in a unified way. However, the formulation in~\cite{LukeLure} considers a zero-crossing reset element whose reset condition is directly associated with the reset-element input. It does not explicitly include the case where the reset instants are determined through an additional shaping filter~\cite{BandpassKarbasi}. Therefore, in this paper,
\begin{itemize}
    \item the Lur'e-based closed-loop HOSIDF formulation is extended to include a reset-triggering signal generated by a shaping filter.
\end{itemize}

The closed-loop HOSIDF-based prediction relies on the practical assumption that the reset action is primarily governed by the first harmonic of the reset-triggering signal~\cite[Assumption~2]{saikumar2021loop}. However, due to the nonlinear nature of the reset element, the reset-triggering signal may contain non-negligible higher-order harmonic components. These components can shift the actual zero crossings of the reset-triggering signal with respect to those predicted by the first harmonic. Consequently, the reset instants used implicitly in the frequency-domain prediction may differ from the reset instants that occur in the actual time-domain response. If these deviations become significant, the predicted closed-loop response may no longer be reliable.

A common practical check in reset-control design is to verify whether the reset-triggering signal exhibits two zero crossings within one period of a sinusoidal input~\cite{ZhangIdentifyTwo,hosseini2025AddOnFilterDesign}. If only two zero crossings occur, the reset behavior is often considered consistent with the first-harmonic-based prediction. However, this check does not quantify the deviation between the predicted and actual reset instants. In particular, the reset-triggering signal may still have two zero crossings per period while the actual reset instants are significantly shifted. To address this limitation,
\begin{itemize}
    \item this paper introduces a frequency-domain metric, denoted by \(\sigma_t\), that quantifies the reset-time deviation between the zero crossings predicted by the first harmonic and those obtained from the reconstructed reset-triggering signal, including higher-order harmonics.
\end{itemize}

Although \(\sigma_t\) measures the deviation of the reset instants, it does not directly indicate how close the reset-triggering signal is to producing additional zero crossings. This is important because multiple zero crossings within one period can lead to unintended reset actions and may invalidate the assumptions used in the closed-loop frequency-domain prediction. Therefore,
\begin{itemize}
    \item the \(\sigma_d\) metric, is introduced to quantify the tendency of the reset-triggering signal to generate additional zero crossings. The metric is constructed from the distance of selected stationary points of the reset-triggering signal to the zero level. In this way, \(\sigma_d\) provides an indication of whether a design is close to violating the two-reset condition, even when the actual number of reset events is still equal to two.
\end{itemize}

The proposed metrics provide diagnostic information that can be used during reset-controller design. In particular, they help identify cases where a controller may have favorable first-harmonic closed-loop behavior but unreliable reset-triggering behavior due to higher-order harmonic content. This motivates the use of shaping filters to reduce the higher-order harmonic contribution in the reset-triggering signal. Shaping filters have previously been used to attenuate undesired nonlinear contributions in reset elements~\cite{BandpassKarbasi}. However, their design may involve higher-order or fractional-order filters. In this work, we consider a simpler first-order shaping filter for a generalized first-order reset element (GFORE). The proposed filter is selected such that, for the considered GFORE structure, the low-frequency attenuation slope of the nonzero HOSIDFs increases. This result is used as a design mechanism to reduce low-frequency higher-order harmonic content in the reset-triggering signal.

The main contributions of this paper are therefore summarized as follows. First, the closed-loop HOSIDF formulation of~\cite{LukeLure} is extended to reset control systems in which the reset-triggering signal is generated through a shaping filter. Second, two frequency-domain metrics are introduced to assess the reliability of the reset instants: \(\sigma_t\), which measures reset-time deviation, and \(\sigma_d\), which measures the tendency toward additional zero crossings. Third, a first-order shaping filter is proposed for a GFORE element to reduce the low-frequency HOSIDF contribution and improve the reset-triggering reliability. Finally, the proposed analysis is evaluated on an industrial wire-bonder motion stage, where the frequency-domain predictions are compared with experimental results.

The remainder of this paper is organized as follows. Section~\ref{sec: priliminaries} introduces the reset element with a shaping filter and recalls the HOSIDF expressions used in the paper. Section~\ref{sec: closed-loop HOSIDF} extends the closed-loop HOSIDF formulation to include the reset-triggering signal. Section~\ref{sec: Metrics} introduces the proposed reliability metrics for reset instants and zero-crossing behavior. Section~\ref{sec: GFORE SF} presents the modified GFORE with a first-order shaping filter. Section~\ref{sec: case study} applies the proposed analysis to the wire-bonder motion stage and discusses the controller design. Section~\ref{sec: Experimental Validation} provides the experimental validation, and Section~\ref{sec: Conclusion} concludes the paper.

\section{Preliminaries}\label{sec: priliminaries}

\subsection{Reset Element}
The closed-loop reset control configuration, shown in Fig.~\ref{Fig: Block diagram CL}, involves exogenous signals 
\(r(t) \in \mathbb{R}\), \(d_i(t)\in \mathbb{R}\), \(d_n(t) \in \mathbb{R}\), and a system output \(y_{\mathrm{p}}(t) \in \mathbb{R}\), 
all considered for \(t \in \mathbb{R}_{\geq 0}\). The plant \(P\) is controlled through a set of LTI filters, namely \(C_{\text{pre}}\), \(C_{\text{par}}\), and \(C_{\text{pos}}\). In addition, the loop contains a reset component, \(\mathcal{R}\), whose dynamics are specified as follows:
\begin{equation}
\label{eq: reset state space}
\mathcal{R} := 
\begin{cases}
    \dot{x}_r(t) = A_r x_r(t) + B_r y(t), & \text{if } \left(x_r(t), q(t)\right) \notin \mathcal{F}, \\[5pt]
    x_r(t^+) = A_\rho x_r(t), & \text{if } \left(x_r(t), q(t)\right) \in \mathcal{F}, \\[5pt]
    u(t) = C_r x_r(t) + D_r y(t), &
\end{cases}
\end{equation}
where the state vector is \( x_r(t) \in \mathbb{R}^{n_r \times 1} \), and the post-reset state is \( x_r(t^{+}) \in \mathbb{R}^{n_r \times 1} \). The state-space matrices of the reset element are given by 
\( A_r \in \mathbb{R}^{n_r \times n_r} \), 
\( B_r \in \mathbb{R}^{n_r \times 1} \), 
\( C_r \in \mathbb{R}^{1 \times n_r} \), 
and \( D_r \in \mathbb{R} \). 
The reset matrix is denoted as 
\( A_\rho = \text{diag}(\gamma_1, \dots, \gamma_{n_r}) \), 
where \( -1 < \gamma_{i} < 1 \) for all \( i \in \mathbb{N} \). 
The signals \( y(t) \in \mathbb{R} \) and \( u(t) \in \mathbb{R} \) represent the input and output of the reset element, respectively.  
Let \(C_s\) represent an LTI shaping filter, which is commonly employed to attenuate undesired components in the reset signal or to shape the nonlinear behaviour of the reset element over a prescribed frequency range of interest~\cite{BandpassKarbasi,ComplexKarbasi}. The state-space realization of \(C_s\) is given by
\begin{equation}
\label{eq: Cs dynamic ss}
C_s :=
\begin{cases}
    \dot{x}_{\mathfrak{q}}(t) = A_{\mathfrak{q}} x_{\mathfrak{q}}(t) + B_{\mathfrak{q}} y(t), \\[5pt]
    q(t) = C_{\mathfrak{q}} x_{\mathfrak{q}}(t) + D_{\mathfrak{q}} y(t),
\end{cases}
\end{equation}
with $x_{\mathfrak{q}}\in \mathbb{R}^{p\times 1}$, \( A_{\mathfrak{q}} \in \mathbb{R}^{p \times p} \), 
\( B_{\mathfrak{q}} \in \mathbb{R}^{p \times 1} \), 
\( C_{\mathfrak{q}} \in \mathbb{R}^{1 \times p} \), 
and \( D_{\mathfrak{q}} \in \mathbb{R} \). The signal \( q(t) \in \mathbb{R} \), referred to as the reset-triggering signal, determines when the state \( x_r \) is reset according to the reset surface \( \mathcal{F} \), defined as  

\begin{equation}
\label{reset surface}
\mathcal{F} := \{ C_{\mathfrak{q}} x_{\mathfrak{q}}(t) + D_{\mathfrak{q}}y(t) = 0 \wedge (A_\rho - I) x_r(t) \neq 0 \}.
\end{equation}

\subsection{(Higher-Order) Sinusoidal Input Describing Functions of Reset Elements}
The inherent nonlinearity of the reset element implies that its steady-state response to a sinusoidal input is non-sinusoidal. Consequently, its frequency response analysis often relies on approximations such as the SIDF method \cite{guo2009frequency}. However, because the SIDF accounts only for the first harmonic of the output while neglecting higher-order components, it may result in significant inaccuracies. To overcome this limitation, HOSIDFs were introduced for reset elements in \cite{saikumar2021loop}, and later extended to reset elements with shaping filters in \cite{ComplexKarbasi}. 

Thus, having the input of the reset element as $y(t)=\hat{y}\sin(\omega t)$, the output $u(t)$ can be described by the Fourier series:
\begin{equation}
    \label{eq: e to y}
u(t)=\sum_{n=1}^{\infty}\left|{H_{n}}(\omega)\right|\hat{y}\sin\left(n\omega t +\angle {H_{n}}(\omega)\right),
\end{equation}
with $n\in \mathbb{N}$, and ${H_{1}}(\omega)$ is the SIDF and ${H_{1}}(\omega),\,\, n>2$ is the HOSIDFs of the reset element with shaping filter, which can be calculated as follows (see \cite[Sec. II.C]{ComplexKarbasi})
\begin{equation}
\label{eq: Hn}
H_{n}(\omega) =
\begin{cases}
C_r(A_r - j\omega I)^{-1}B_r\Theta_{\varphi}(\omega)\\
\qquad+ C_r(j\omega I - A_r)^{-1}B_r + D_r, & n=1, \\[6pt]
C_r(A_r - j\omega n I)^{-1}B_r\Theta_{\varphi}(\omega), & \text{odd } n > 2, \\[6pt]
0, & \text{even } n \ge 2,
\end{cases} \\[6pt]
\end{equation}
with
\begin{align}
\label{eq: HOSIDFs all functions}
\Theta_{\varphi}(\omega) &=
\frac{-2j\omega e^{nj\varphi(\omega)}}{\pi}\,
\Omega(\omega)\Upsilon(\omega)\Lambda^{-1}(\omega) \nonumber\\[0pt]
\Upsilon(\omega)&=\omega I \cos\left(\varphi(\omega)\right) - A_r \sin\left(\varphi(\omega)\right)\nonumber\\
\Omega(\omega) &= \Delta(\omega) - \Delta(\omega)\Delta_{\rho}^{-1}(\omega) A_{\rho} \Delta(\omega)\nonumber \\[0pt]
\Lambda(\omega) &= \omega^{2}I + A_{r}^{2}\nonumber \\[0pt]
\Delta(\omega) &= I + e^{\tfrac{\pi}{\omega}A_{r}}\nonumber \\[0pt]
\Delta_{\rho}(\omega) &= I + A_{\rho} e^{\tfrac{\pi}{\omega}A_{r}}.\nonumber\\
& 
\end{align}
In the above equations, $\varphi(\omega) = \arg\left(C_s(j\omega)\right)$ denotes the phase shift in the reset signal relative to the case without a shaping filter ($C_s=1$). In the absence of a shaping filter, the reset occurs at the instants $t_k = \tfrac{k\pi}{\omega}$, where $\hat{y} \sin(\omega t) = 0$. With the shaping filter included, however, the reset occurs at $t_k = \tfrac{k\pi - \varphi}{\omega}$, where $q(t)=\hat{y} |C_s(j\omega)| \sin\left(\omega t + \varphi(\omega)\right) = 0$.

A reset element follows its base linear system (BLS) dynamics if no reset happens $\left((x_r(t), q(t)) \notin \mathcal{F}\,\,\, \forall\,\,t \in \mathbb{R}_{\geq 0}\right)$. Thus, the transfer function of its BLS is defined as follows
\begin{equation}
\label{eq RCS bls}
R_{\mathrm{bl}}(s) = C_r(sI - A_r)^{-1}B_r + D_r,
\end{equation}
where $s \in\mathbb{C}$ is the Laplace variable.

\usetikzlibrary {arrows.meta}
\tikzstyle{block} = [draw,thick, fill=white, rectangle, minimum height=2em, minimum width=2.5em, anchor=center]
\tikzstyle{sum} = [draw, fill=white, circle, minimum height=0.6em, minimum width=0.6em, anchor=center, inner sep=0pt]
\usetikzlibrary {arrows.meta}
\tikzstyle{block} = [draw,thick, fill=white, rectangle, minimum height=1.75em, minimum width=2.1875em, anchor=center]
\tikzstyle{block2} = [draw,thick, fill=white, rectangle, minimum height=0.8em, minimum width=1.0em, anchor=center]
\tikzstyle{sum} = [draw, fill=white, circle, minimum height=0.6em, minimum width=0.6em, anchor=center, inner sep=0pt]
\begin{figure}[t]
	\centering
	\begin{scaletikzpicturetowidth}{\linewidth}
		\begin{tikzpicture}[scale=\tikzscale]
			\node[coordinate](input) at (0,0) {};
			\node[sum] (sum1) at (1,0) {};
			\node[sum] (sum3) at (8.75,0) {+};
			\node[sum] (sum4) at (11,-2) {+};
			\node[sum, fill=black, minimum size=0.4em] (dot2) at (11,0) {};	
            \node[sum, fill=black, minimum size=0.0em] (dotCs) at (3.3,1.0) {};
            \node[sum, fill=black, minimum size=0.0em] (dotImag) at (4.4,1.7) {};
            \node[sum, fill=black, minimum size=0.0em] (dotImag2) at (4.88,1.7) {};
			\node[sum] (sumN) at (6.10,0) {+};
			\node[block] (lead) at (5.4,2.0) {$\mathcal{R}$};
            \node[block2] (Cs) at (3.9,1.0) {$C_{s}$};
\node[block] (C_par) at (4.75,0) {$C_{\text{par}}$};
            
			\node[block] (controller) at (7.3,0) {$C_{\text{pos}}$};
			\node[block] (fo-higs) at (2.18,0) {$C_{\text{pre}}$};
			\node[block] (system) at (10,0) {$P$};

            \node[sum, fill=black, minimum size=0.4em] (Dot_1) at (3.3,0) {};
			\node[coordinate](output) at (12,0) {};
			\node[coordinate](di-input) at (8.75,1) {};
			\node[coordinate](n-input) at (12,-2) {};
			\draw[arrows = {-Latex[width=6pt, length=6pt]}] (input)  -- node[above]{$r$} (sum1);
			\draw[arrows = {-Latex[width=6pt, length=6pt]}] (di-input)node[above]{$d_i$}  --  (sum3);
			\draw[arrows = {-Latex[width=6pt, length=6pt]}] (n-input)  -- node[above]{$d_n$} (sum4);
			\draw[arrows = {-Latex[width=6pt, length=6pt]}] (sum1)   --node[above]{$e$}  (fo-higs);
			\draw[ = {-Latex[width=6pt, length=6pt]}] (fo-higs) --node [above]{$\,y$} (Dot_1);
            \draw[ = {-Latex[width=6pt, length=6pt]}, dash pattern=on 1.1pt off 0.9pt] (dotCs) --node [above]{} (Cs);
            \draw[  = {-Latex[width=6pt, length=6pt]}, dash pattern=on 1.1pt off 0.9pt] (Cs) -|node [above]{} (dotImag);
            \draw[arrows={-Latex[width=4pt, length=4pt]}, dash pattern=on 1.1pt off 0.9pt] 
  (dotImag) -- node [below]{$q$}(dotImag2);
			\draw[arrows = {-Latex[width=6pt, length=6pt]}] (lead)  -| (sumN) node[above, xshift=+6pt,yshift=+10pt] {$u$};

            \draw[arrows = {-Latex[width=6pt, length=6pt]}] (C_par)  --node[above] {} (sumN);
            
			\draw[arrows = {-Latex[width=6pt, length=6pt]}] (controller)  --node[above] {} (sum3);
			\draw[arrows = {-Latex[width=6pt, length=6pt]}] (sum3) -- (system);
			\draw[arrows = {-Latex[width=6pt, length=6pt]}] (system)  -- node[above]{$y_{\mathrm{p}}$} (output);
			\draw[arrows = {-Latex[width=6pt, length=6pt]}] (dot2) -- (sum4);

\draw[arrows = {-Latex[width=6pt, length=6pt]}] (sumN) -- (controller);
            
			\draw[arrows = {-Latex[width=6pt, length=6pt]}] (sum4) -| node[pos=0.9,left]{$-$} (sum1);

            \draw[arrows = {-Latex[width=6pt, length=6pt]}] (Dot_1) |-node[pos=0.85,left]{} (lead);

            \draw[arrows = {-Latex[width=6pt, length=6pt]}] (Dot_1) |-node[pos=0.85,left]{} (C_par);

            \draw[black,line width=0.6pt]
    ($(lead.south west)+(+0.085,0.09)$)
    -- ++(0.07,0)
    -- ++(0,0.25)
    -- ++(0.12,0)
    -- ++(0,-0.25)
    -- ++(0.07,0);
			
		\end{tikzpicture}
	\end{scaletikzpicturetowidth}
	\caption{Block diagram of the closed-loop reset control system.}
	\label{Fig: Block diagram CL}	
\end{figure}
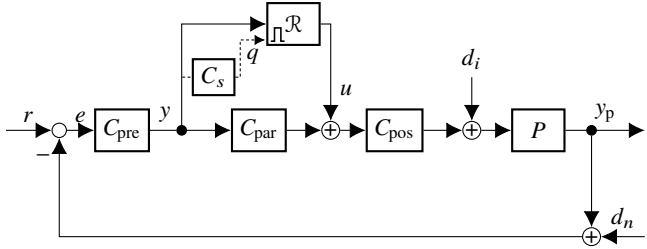

\section{Closed-loop Frequency Response in the Presence of a Shaping Filter}
\label{sec: closed-loop HOSIDF}

Due to the nonlinear element in the feedback loop, closed-loop frequency-domain functions, such as the sensitivity and complementary sensitivity functions, no longer satisfy the conventional open-loop/closed-loop relations of LTI systems. Analytical frequency-response expressions for closed-loop reset control systems were introduced in~\cite{saikumar2021loop}. This framework was later generalized in~\cite{LukeLure} to a broader class of reset control structures, including configurations with pre-, post-, and parallel filters, as shown in Fig.~\ref{Fig: Block diagram CL}. However, the formulation in~\cite{LukeLure} does not explicitly account for the presence of a shaping filter in the reset-triggering path.

In this work, the reset-triggering signal \(q(t)\) is studied using frequency-domain tools in order to assess the nonlinear behavior and reliability of the reset action. In particular, we focus on cases where \(q(t)\) exhibits a strong nonlinear contribution, which may lead to undesirable or unreliable reset instants. To mitigate this effect, a shaping-filter-based approach is proposed. Therefore, we first extend the closed-loop frequency-response framework of~\cite{LukeLure} to structures that include the shaping filter \(C_s\).

Following~\cite{LukeLure}, all LTI components of the closed-loop system are collected into a single LTI block, which is placed in feedback with the reset element. This representation is closely related to the well-known Lur'e configuration~\cite{khalil2002nonlinear}. Since any interconnection of LTI elements with a single reset element can be rewritten in this form, it provides a convenient and general framework for closed-loop frequency-domain analysis. The corresponding Lur'e-type representation of the closed-loop system in Fig.~\ref{Fig: Block diagram CL} is shown in Fig.~\ref{fig: Lure CL}.

In this representation, the LTI block \(G\) is placed in feedback with the reset element \(\mathcal{R}\). The external input \(w\in\mathbb{R}\) and the reset-element output \(u\) are the two inputs of \(G\). The signal \(z\in\mathbb{R}\) denotes a generic performance output of interest, such as the tracking error, control input, or plant output. Furthermore, \(y(t)\) denotes the input to the reset element, as defined in~\eqref{eq: reset state space}, whereas \(q(t)\) denotes the reset-triggering signal generated through the shaping filter. Hence, in contrast to the formulation without a shaping filter, the reset element receives \(y(t)\), while its reset instants are determined by the zero crossings of \(q(t)\).

The input-output relation of the LTI block can therefore be written as
\begin{equation}
\begin{bmatrix}
Z(s)\\
Y(s)\\
Q(s)
\end{bmatrix}
=
\hat{G}(s)
\begin{bmatrix}
W(s)\\
U(s)
\end{bmatrix},
\label{eq: G_IO_shaping}
\end{equation}
where
\begin{equation}
\hat{G}(s):=
\begin{bmatrix}
G(s)\\
G_{wq}(s) \quad G_{uq}(s)
\end{bmatrix},
\qquad
G(s):=
\begin{bmatrix}
G_{wz}(s) & G_{uz}(s)\\
G_{wy}(s) & G_{uy}(s)
\end{bmatrix}.
\label{eq: Gprime_matrix_shaping}
\end{equation}
Here, \(G(s)\) corresponds to the original LTI block relating \(W(s)\) and \(U(s)\) to \(Z(s)\) and \(Y(s)\), whereas the additional row in \(\hat{G}(s)\) describes the transfer paths from \(W(s)\) and \(U(s)\) to the reset-triggering signal \(Q(s)\). Hence, the shaping filter is embedded in \(\hat{G}(s)\) through \(G_{wq}(s)\) and \(G_{uq}(s)\). When \(C_s(s)=1\), the reset-triggering signal coincides with the reset-element input, i.e., \(q(t)=y(t)\), and the formulation reduces to the standard zero-crossing reset configuration considered in~\cite{LukeLure}. The extended representation in~\eqref{eq: G_IO_shaping}-\eqref{eq: Gprime_matrix_shaping} enables a direct extension of the closed-loop HOSIDF framework in~\cite{LukeLure}. In particular, the closed-loop frequency-domain relations are still derived through the reset-element input \(y(t)\), while the additional output channel \(q(t)\) enables an explicit frequency-domain characterization of the reset-triggering signal. Hence, besides predicting the performance output \(z(t)\), the proposed formulation also characterizes the harmonic content of \(q(t)\), which is central to the reliability analysis considered in this work. The corresponding closed-loop HOSIDFs are given next.

\begin{theorem}
\label{thm:closed_loop_shaping}
Consider the closed-loop reset control system in Fig.~\ref{fig: Lure CL}, where the extended LTI block $\hat{G}$ is defined in~\eqref{eq: G_IO_shaping}--\eqref{eq: Gprime_matrix_shaping}. Let the system be subject to a sinusoidal input
\begin{equation}
w(t)=\hat{w}\sin(\omega t+\phi_w),
\end{equation}
with $\hat{w}\in\mathbb{R}{>0}$, $\phi_w\in\mathbb{R}$, and $\omega\in\mathbb{R}{>0}$. Assume that the closed-loop system is sinusoidal-input convergent, such that $u(t)$, $y(t)$, $q(t)$, and $z(t)$ converge to unique zero-mean periodic solutions with period $2\pi/\omega$. Furthermore, assume that the reset element with shaping filter is characterized by the SIDF/HOSIDFs $H_n(\omega)$, as defined in~\eqref{eq: Hn}, and that the closed-loop harmonic components are evaluated under the harmonic-separator approximation used in HOSIDF-based closed-loop analysis. Then, the HOSIDF-based approximation of the reset-element input is given by
\begin{equation}
y(t)
=
\hat{w}
\sum_{n=1}^{\infty}
|S_{y,n}(\omega)|
\sin\!\left(n(\omega t+\phi_w)+\angle S_{y,n}(\omega)\right),
\label{eq: y_periodic_shaping}
\end{equation}
where
\begin{equation}
S_{y,1}(\omega)
:=
\frac{G_{wy}(j\omega)}
{1-G_{uy}(j\omega)H_1(\omega)},
\label{eq: Sy1_shaping}
\end{equation}
\begin{equation}
S_{y,n}(\omega)
:=
\frac{
G_{uy}(nj\omega)H_n(\omega)S_{y,1}(\omega)
e^{j(n-1)\angle S_{y,1}(\omega)}
}
{1-G_{uy}(nj\omega)R_{\mathrm{bl}}(nj\omega)}\,\,\,\forall n\neq 1.
\label{eq: Syn_shaping}
\end{equation}
Moreover, the HOSIDF-based approximation of the reset-triggering signal is given by
\begin{equation}
q(t)
=
\hat{w}
\sum_{n=1}^{\infty}
|S_{q,n}(\omega)|
\sin\!\left(n(\omega t+\phi_w)+\angle S_{q,n}(\omega)\right),
\label{eq: q_periodic_shaping}
\end{equation}
with
\begin{equation}
S_{q,1}(\omega)
:=
G_{wq}(j\omega)
+
G_{uq}(j\omega)H_1(\omega)S_{y,1}(\omega),
\label{eq: Sq1_shaping}
\end{equation}
\begin{equation}
S_{q,n}(\omega)
:=
\frac{
G_{uq}(nj\omega)H_n(\omega)S_{y,1}(\omega)
e^{j(n-1)\angle S_{y,1}(\omega)}
}
{1-G_{uy}(nj\omega)R_{\mathrm{bl}}(nj\omega)}\,\,\,\forall n\neq 1.
\label{eq: Sqn_shaping}
\end{equation}
Finally, the HOSIDF-based approximation of the performance output is given by
\begin{equation}
z(t)
=
\hat{w}
\sum_{n=1}^{\infty}
|S_{z,n}(\omega)|
\sin\!\left(n(\omega t+\phi_w)+\angle S_{z,n}(\omega)\right),
\label{eq: z_periodic_shaping}
\end{equation}
where
\begin{equation}
S_{z,1}(\omega)
:=
G_{wz}(j\omega)
+
G_{uz}(j\omega)H_1(\omega)S_{y,1}(\omega),
\label{eq: Sz1_shaping}
\end{equation}
\begin{equation}
S_{z,n}(\omega)
:=
\frac{
G_{uz}(nj\omega)H_n(\omega)S_{y,1}(\omega)
e^{j(n-1)\angle S_{y,1}(\omega)}
}
{1-G_{uy}(nj\omega)R_{\mathrm{bl}}(nj\omega)}\,\,\,\forall n\neq 1.
\label{eq: Szn_shaping}
\end{equation}
\end{theorem}
\textbf{Proof:} See Appendix~\ref{sec: App A}.

Theorem~\ref{thm:closed_loop_shaping} is adopted from the closed-loop HOSIDF framework presented in~\cite{LukeLure} and is generalized here for reset control systems that include shaping filters. In this extension, the closed-loop frequency-domain relations remain structurally the same as in~\cite{LukeLure}, while the SIDF/HOSIDFs \(H_n(\omega)\) are replaced by those corresponding to the reset element with the shaping filter. Furthermore, the reset-triggering signal \(q(t)\) is explicitly included as an additional output channel of the extended LTI block \(\hat{G}(s)\). For a detailed discussion of the assumptions used in the theorem, including sinusoidal-input convergence, the harmonic-separator approximation, and the general interpretation of closed-loop HOSIDF-based frequency-response analysis, the reader is referred to~\cite{LukeLure}.

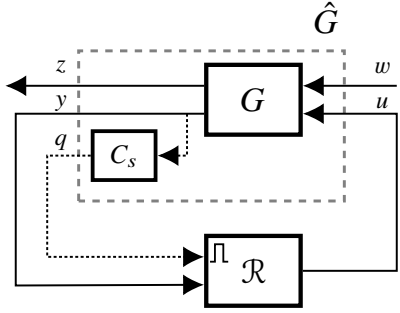
\begin{figure}[!t]
\centering
\resizebox{0.30\textwidth}{!}{%
\begin{tikzpicture}[
     x=1cm,y=1cm,
    every node/.style={font=\large},
    wire/.style={draw=black,line width=0.8pt,line cap=butt,line join=miter},
   sig/.style={wire,-{Latex[length=3.0mm,width=3mm]}},
    block/.style={
        draw=blockblue,
        line width=1.75pt,
        fill=white,
        minimum width=1.35cm,
        minimum height=1.00cm,
        inner sep=0pt,
        font=\Large
    },
    smallblock/.style={
        draw=blockblue,
        line width=1.55pt,
        fill=white,
        minimum width=0.90cm,
        minimum height=0.70cm,
        inner sep=0pt,
        font=\large
    },
    hatbox/.style={
        draw=blockblue,
        dashed,
        line width=1.20pt
    }
]

\node[block]      (G)  at (3.15, 2.20) {\huge{$G$}};
\node[smallblock] (Cs) at (1.3, 1.4) {\Large{$C_s$}};
\node[block]      (R)  at (3.15,-0.25) {\huge{$\mathcal{R}$}};
\draw[black,line width=0.75pt]
    ($(R.north west)+(+0.085,-0.4)$)
    -- ++(0.07,0)
    -- ++(0,0.25)
    -- ++(0.12,0)
    -- ++(0,-0.25)
    -- ++(0.07,0);

\draw[hatbox,gray] (0.65,0.75) rectangle (4.45,2.9);
\node[font=\Large] at (4.2,3.37) {\huge{$\hat{G}$}};

\coordinate (Gz) at ($(G.west)+(0, 0.20)$);
\coordinate (Gy) at ($(G.west)+(0,-0.20)$);
\coordinate (Gw) at ($(G.east)+(0, 0.20)$);
\coordinate (Gu) at ($(G.east)+(0,-0.20)$);

\coordinate (Ru)   at ($(R.west)+(0, 0.20)$);
\coordinate (Ry)   at ($(R.west)+(0,-0.20)$);
\coordinate (Rout) at (R.east);

\coordinate (zout)    at (-0.4, 2.40);
\coordinate (win)     at ( 5.2, 2.40);
\coordinate (ubus)    at ( 5.20, 2.00);
\coordinate (yleft)   at (-0.25, 2.00);
\coordinate (ybottom) at (-0.25,-0.46);

\coordinate (csinup)  at (2.2,2);
\coordinate (csin)    at (2.2,1.4);
\coordinate (vleft)   at (0.20,1.4);
\coordinate (vdown)   at (0.20,-0.05);


\draw[sig] (Gz) -- (zout);
\draw[sig] (win) -- (Gw);

\node[above=2pt] at (0.4,2.40) {\large{$z$}};
\node[above=2pt] at ( 5.0,2.40) {\large{$w$}};

\draw[sig] (Gy) -- (yleft) -- (ybottom) -- (Ry);
\node[above=2pt] at (0.4,1.85) {\large{$y$}};

\draw[sig,dash pattern=on 1.2pt off 1.0pt] (csinup) -- (csin) -- (Cs.east);
\draw[sig,dash pattern=on 1.2pt off 1.0pt] (Cs.west) -- (vleft) -- (vdown) -- (Ru);

\node[above=2pt] at (0.4,1.3) {\large{$q$}};

\draw[sig] (Rout) -- (5.20,-0.25) -- (ubus) -- (Gu);
\node[above=2pt] at (5.0,1.9) {\large{$u$}};

\end{tikzpicture}%
}
\caption{Reset structure with the shaping filter $C_s$ included in the augmented plant $\hat{G}$.}
\label{fig: Lure CL}
\end{figure}
\section{Reliability Metrics for Reset-Triggering Signal}
\label{sec: Metrics}

In the previous section, the closed-loop frequency-domain response of reset control systems in the presence of a shaping filter was derived. The connection between the HOSIDFs of the reset element and the corresponding closed-loop HOSIDFs was first introduced in~\cite{saikumar2021loop} using a harmonic-separation approximation. In this approximation, only the first harmonic component of the reset-element input,
\begin{equation}
y_1(t)
=
\hat{w}
|S_{y,1}(\omega)|
\sin\!\left(\omega t+\phi_w+\angle S_{y,1}(\omega)\right),
\label{eq: y1}
\end{equation}
is assumed to enter the nonlinear reset mechanism; see also~\cite[Fig.~1]{LukeLure}. Consequently, in the considered architecture with a shaping filter, only the first harmonic component of the reset-triggering signal,
\begin{equation}
q_1(t)
=
\hat{w}
|S_{q,1}(\omega)|
\sin\!\left(\omega t+\phi_w+\angle S_{q,1}(\omega)\right),
\label{eq: q1}
\end{equation}
is assumed to determine the reset instants.

However, the actual reset-triggering signal is \(q(t)\), which also contains higher-order harmonic components. From~\eqref{eq: q1}, \(q_1(t)\) is a pure sinusoidal signal and therefore has exactly two zero crossings over one fundamental period \(2\pi/\omega\), provided that \(|S_{q,1}(\omega)|\neq 0\). Hence, the closed-loop frequency-domain prediction is derived based on the reset instants associated with these two zero crossings. If the higher-order harmonics in \(q(t)\) are sufficiently small, the actual reset instants remain close to those predicted from \(q_1(t)\), and the closed-loop HOSIDF-based prediction is expected to remain reliable.

First, let \(N_r\) denote the number of reset events occurring within one period of the sinusoidal input. This is defined as the number of distinct zero crossings of the reset-triggering signal \(q(t)\) over the interval \(t\in[0,2\pi/\omega)\), i.e.,
\begin{equation}
\label{eq: nr}
N_r
:=
\#\left\{
t_k\in\left[0,\frac{2\pi}{\omega}\right)
:
q(t_k,\omega)=0
\right\},
\qquad
N_r\in\mathbb{N},\quad N_r\geq 2.
\end{equation}
For the first-harmonic approximation \(q_1(t)\), it follows that \(N_r=2\), provided that \(|S_{q,1}(\omega)|\neq 0\). Therefore, the closed-loop frequency-domain prediction is based on the assumption that two reset events occur within each fundamental period.

In contrast, two scenarios may reduce the reliability of both the frequency-domain performance prediction and the designed reset controller itself. First, the actual reset-triggering signal \(q(t)\) may still satisfy \(N_r=2\), but the corresponding reset instants may deviate from the zero crossings predicted by \(q_1(t)\). In this case, the number of reset events remains consistent with the assumed model, but the reset instants are shifted. Second, the higher-order harmonics may significantly alter the shape of \(q(t)\), resulting in \(N_r>2\). This violates the two-reset condition implicitly used in the closed-loop frequency-domain prediction, potentially leading to unreliable or degraded performance. The following subsections analyze these two scenarios and introduce metrics to quantify the corresponding reset-instant reliability.

\subsection{Reliability of Reset Instants}
\label{subsec: reset_time_deviation}

Using the closed-loop representation of \(q(t)\) in~\eqref{eq: q_periodic_shaping}, the reset-triggering signal can be reconstructed at each excitation frequency \(\omega\). Since, under the considered HOSIDF representation, the reset-triggering signal contains only odd harmonics, the signal has half-wave symmetry. Therefore, it is sufficient to analyze the deviation around one of the two nominal reset instants within a period; the same analysis applies to the other reset instant shifted by \(T/2\), where \(T=2\pi/\omega\).

Let \(t_k\) denote a nominal reset instant predicted by the first harmonic \(q_1(t)\), i.e.,
\begin{equation}
q_1(t_k)=0.
\end{equation}
Equivalently, for
\[
q_1(t)
=
\hat{w}|S_{q,1}(\omega)|
\sin\!\left(\omega t+\phi_w+\angle S_{q,1}(\omega)\right),
\]
the nominal reset instants are given by
\begin{equation}
t_k
=
\frac{k\pi-\phi_w-\angle S_{q,1}(\omega)}{\omega},
\qquad k\in\mathbb{N}.
\label{eq: nominal_reset_time}
\end{equation}
In contrast, the actual reset instants are determined by the zero crossings of the full reset-triggering signal \(q(t)\). Around the nominal reset instant \(t_k\), define the set of actual reset instants as
\begin{equation}
\mathcal{T}_k^{*}(\omega)
:=
\left\{
t_{k,i}^{*}\in
\left[t_k-\frac{T}{4},\,t_k+\frac{T}{4}\right)
:
q(t_{k,i}^{*},\omega)=0
\right\}.
\label{eq: actual_reset_set}
\end{equation}
If \(N_r=2\), then \(\mathcal{T}_k^{*}(\omega)\) contains one element, denoted by \(t_{k,1}^{*}\). If \(N_r=2p\) with \(p>1\), then multiple zero crossings occur around each nominal reset instant, and we write
\begin{equation}
\mathcal{T}_k^{*}(\omega)
=
\left\{
t_{k,1}^{*},t_{k,2}^{*},\ldots,t_{k,p}^{*}
\right\}.
\end{equation}

\begin{definition}
\label{Def: sigma_t}
The reset-time deviation metric is defined as
\begin{equation}
\label{eq: sigma_t}
\sigma_t(\omega)
=
\frac{|\Delta_t(\omega)|}{T/2}\times 100 .
\end{equation}
For the case \(N_r=2\), where only one actual reset instant exists around \(t_k\), the deviation is given by
\begin{equation}
\Delta_t(\omega)
=
t_k-t_{k,1}^{*}.
\label{eq: delta_t_two_resets}
\end{equation}
For the case \(N_r>2\), the deviation is defined as the maximum spread between the nominal reset instant and the actual reset instants around it, namely
\begin{equation}
\Delta_t(\omega)
=
\max\!\left(\{t_k\}\cup\mathcal{T}_k^{*}(\omega)\right)
-
\min\!\left(\{t_k\}\cup\mathcal{T}_k^{*}(\omega)\right).
\label{eq: delta_t_multiple_resets}
\end{equation}
\end{definition}

The normalization by \(T/2\) expresses the deviation relative to the nominal time interval between two consecutive reset events predicted by \(q_1(t)\). Hence, for \(N_r=2\), \(\sigma_t(\omega)\) quantifies the shift between the predicted and actual reset instants. For \(N_r>2\), it quantifies the spread of the multiple actual reset instants around the nominal one. Therefore, \(\sigma_t(\omega)\) should be interpreted together with \(N_r\): a small value of \(\sigma_t(\omega)\) indicates that the actual reset action remains close to the assumed timing, whereas a large value indicates a significant deviation from the reset instants used in the closed-loop frequency-domain prediction. 

To illustrate how \(\sigma_t(\omega)\) captures deviations in the reset instants, Fig.~\ref{fig: scenarios} shows three representative cases in which the actual reset-triggering signal \(q(t)\) deviates from its first harmonic approximation \(q_1(t)\). These examples are only intended for visualization and are not associated with a specific reset control system. Fig.~\ref{fig: q_case_1} represents a case where the higher-order harmonic content is negligible, and the zero crossings of \(q(t)\) remain close to those of \(q_1(t)\). Fig.~\ref{fig: q_case_2} shows a case with stronger higher-order harmonic content, where the actual reset instants are shifted, while the number of reset events remains \(N_r=2\). Therefore, in both cases, \(\Delta_t\) can be computed according to~\eqref{eq: delta_t_two_resets}. In contrast, Fig.~\ref{fig: q_case_3} illustrates a case where the higher-order harmonics are sufficiently large to introduce multiple zero crossings. In this case, \(N_r=6\), meaning that three actual zero crossings of \(q(t)\) occur around each nominal zero crossing of \(q_1(t)\). Consequently, \(\Delta_t\) can be computed using~\eqref{eq: delta_t_multiple_resets}, which accounts for the worst-case spread between the nominal and actual reset instants.

\begin{figure}[!t]
\centering
\subfloat[]{%
\includegraphics[width=0.8\columnwidth]{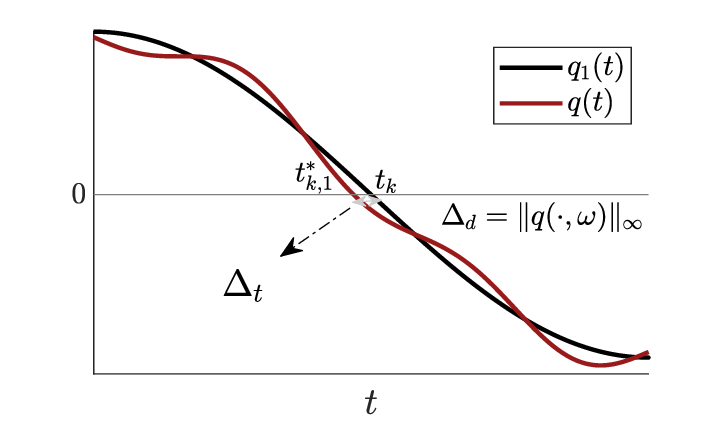}%
\label{fig: q_case_1}}
\hfil
\subfloat[]{%
\includegraphics[width=0.8\columnwidth]{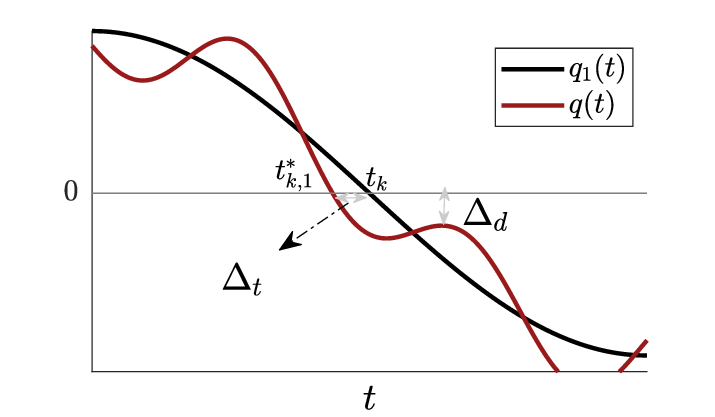}%
\label{fig: q_case_2}}
\hfil
\subfloat[]{%
\includegraphics[width=0.8\columnwidth]{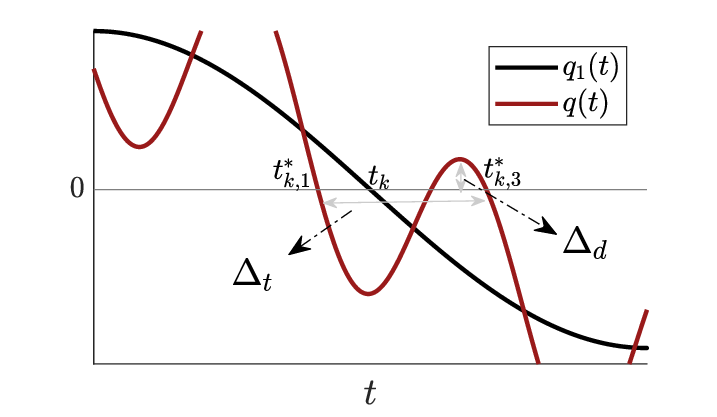}%
\label{fig: q_case_3}}
\caption{Reset-triggering signal \(q(t)\) and its first harmonic approximation \(q_1(t)\):  
(a) negligible higher-order harmonic content,  
(b) significant higher-order harmonic content without additional zero crossings, i.e., \(N_r=2\),  
(c) strong higher-order harmonic content resulting in multiple zero crossings, i.e., \(N_r=6\).}
\label{fig: scenarios}
\end{figure}

\subsection{Reliability of the Number of Zero Crossings}
\label{subsec: reliability_zero_crossings}

The condition \(N_r>2\) indicates a direct violation of the two-reset assumption used in the closed-loop frequency-domain prediction. In this case, the reset element is no longer triggered according to the zero crossings implied by the first-harmonic approximation \(q_1(t)\). Instead, the reset instants are determined by the full reset-triggering signal \(q(t)\), whose higher-order harmonic content has become sufficiently dominant to introduce additional zero crossings; see Fig.~\ref{fig: q_case_3}. The number of zero crossings has therefore been used in the literature as a practical indicator for evaluating whether a reset controller is properly designed~\cite{ZhangIdentifyTwo,hosseini2025AddOnFilterDesign}.

However, satisfying \(N_r=2\) does not necessarily imply reliable reset behavior. As illustrated in Fig.~\ref{fig: q_case_2}, the reset-triggering signal may still have only two zero crossings per period while containing significant higher-order harmonic content. In such cases, \(q(t)\) can be close to producing additional zero crossings, and a small increase in the nonlinear contribution may lead to the situation shown in Fig.~\ref{fig: q_case_3}. This is undesirable not only because it reduces the reliability of the predicted closed-loop frequency-domain response, but also because multiple unintended reset events can significantly degrade the actual closed-loop performance. This behavior is conceptually similar to the gain-loss phenomenon observed in HIGS-based systems when excessive switching between modes occurs~\cite{GainLossHIGS}.

Therefore, besides checking whether \(N_r=2\), it is useful to quantify how close the reset-triggering signal is to producing additional zero crossings. To this end, a zero-crossing reliability metric \(\sigma_d(\omega)\) is introduced. Let
\begin{equation}
\|q(\cdot,\omega)\|_{\infty}
:=
\max_{t\in[0,T)}
|q(t,\omega)|,
\qquad
T=\frac{2\pi}{\omega}.
\label{eq: q_inf_norm}
\end{equation}
For the case \(N_r=2\), define the set of critical points as
\begin{equation}
\mathcal{E}_{2}(\omega)
:=
\left\{
t\in[0,T):
\dot{q}(t,\omega)=0,\;
q(t,\omega)\ddot{q}(t,\omega)>0
\right\}.
\label{eq: E2_set}
\end{equation}
The points in \(\mathcal{E}_{2}(\omega)\) correspond to local extrema directed toward the zero level. Therefore, their distance to zero indicates how close \(q(t)\) is to generating additional zero crossings while still satisfying \(N_r=2\); see Fig.~\ref{fig: q_case_2}.

For the case \(N_r>2\), the additional zero crossings are evaluated locally around each nominal reset instant \(t_k\). Let
\begin{equation}
\mathcal{I}_k(\omega)
:=
\left[
\min\!\left(\{t_k\}\cup\mathcal{T}_k^{*}(\omega)\right),
\max\!\left(\{t_k\}\cup\mathcal{T}_k^{*}(\omega)\right)
\right],
\label{eq: Ik_interval}
\end{equation}
where \(\mathcal{T}_k^{*}(\omega)\) denotes the set of actual zero crossings of \(q(t)\) associated with the nominal reset instant \(t_k\), as defined in~\eqref{eq: actual_reset_set}. Then, the relevant critical points associated with the additional zero-crossing loop are defined as
\begin{equation}
\label{eq: Egt2_set}
\begin{aligned}
\mathcal{E}_{>2,k}(\omega)
:=
\Bigl\{
t\in\mathcal{I}_k(\omega)\; \big|\;
&\dot{q}(t,\omega)=0 \wedge
\big(q(t,\omega)\ddot{q}(t,\omega)<0 ,\\
&
\vee q(t,\omega)=0\big) \wedge \dot{q}_1(t,\omega)\ddot{q}(t,\omega)>0
\Bigr\}.
\end{aligned}
\end{equation}
The restriction to \(\mathcal{I}_k(\omega)\) ensures that only the extrema associated with the additional zero crossings around the nominal reset instant \(t_k\) are considered, while unrelated extrema elsewhere in the period are excluded; see Fig.~\ref{fig: q_case_3}. The sign conditions used in \(\mathcal{E}_{2}(\omega)\) and \(\mathcal{E}_{>2,k}(\omega)\) are further discussed in Appendix~\ref{app: stationary_point_conditions}.

The distance measure \(\Delta_d(\omega)\) is then defined as
\begin{equation}
\Delta_d(\omega)
=
\begin{cases}
\displaystyle
\min_{t\in\mathcal{E}_{2}(\omega)}
|q(t,\omega)|, 
& N_r=2,\;\mathcal{E}_{2}(\omega)\neq\emptyset, \\[8pt]
\displaystyle
\|q(\cdot,\omega)\|_{\infty}, 
& N_r=2,\;\mathcal{E}_{2}(\omega)=\emptyset, \\[8pt]
\displaystyle
\max_{t\in\mathcal{E}_{>2,k}(\omega)}
|q(t,\omega)|, 
& N_r>2.
\end{cases}
\label{eq: Delta_d}
\end{equation}
For cases such as Fig.~\ref{fig: q_case_1}, where no critical stationary point exists for \(N_r=2\), the convention \(\Delta_d(\omega)=\|q(\cdot,\omega)\|_{\infty}\) is used. This corresponds to the most reliable case in the proposed metric.

\begin{definition}
\label{Def: sigma_d}
The zero-crossing reliability metric is defined as
\begin{equation}
\sigma_d(\omega)
=
\begin{cases}
\displaystyle
\left(
1-
\frac{\Delta_d(\omega)}
{\|q(\cdot,\omega)\|_{\infty}}
\right)\times 100, & N_r=2, \\[12pt]
\displaystyle
\left(
1+
\frac{\Delta_d(\omega)}
{\|q(\cdot,\omega)\|_{\infty}}
\right)\times 100, & N_r>2.
\end{cases}
\label{eq: sigma_d}
\end{equation}
\end{definition}

For \(N_r=2\), the metric satisfies \(0\leq\sigma_d(\omega)<100\). A small value indicates that the reset-triggering signal is far from producing additional zero crossings, whereas a value close to \(100\) indicates that \(q(t)\) is close to touching the zero level at an additional point. Therefore, even though the two-reset condition is not yet violated, the reset behavior is close to becoming unreliable. For \(N_r>2\), the metric satisfies \(\sigma_d(\omega)>100\), explicitly indicating that additional zero crossings have already occurred. In this case, larger values of \(\sigma_d(\omega)\) correspond to a stronger violation of the two-reset condition.

In the numerical evaluation, the reconstructed signal \(q(t,\omega)\) is obtained by truncating the harmonic expansion in \eqref{eq: q_periodic_shaping} up to \(n=n_{\max}\). The zero crossings and stationary points are then computed over one fundamental period using a dense time grid, followed by local root refinement. Increasing \(n_{\max}\) beyond 1000 did not change the reported values of \(N_r\), \(\sigma_t\), or \(\sigma_d\) over the considered frequency range.

The metrics introduced in this section provide a frequency-domain means to assess whether the reset-triggering signal remains consistent with the assumptions underlying the closed-loop HOSIDF prediction. In particular, large values of \(\sigma_t(\omega)\) or \(\sigma_d(\omega)\) indicate that the higher-order harmonic content of \(q(t)\) may lead to shifted or additional reset events. These observations motivate the design of shaping filters that directly reduce the low-frequency contribution of the HOSIDFs and thereby improve the reliability of the reset-triggering mechanism. In the next section, this idea is developed for a modified GFORE equipped with a first-order shaping filter.
\section{Modified GFORE With a First-Order Shaping Filter}
\label{sec: GFORE SF}

As discussed earlier, the shaping filter \(C_s\) can attenuate undesired components in the reset-element input \(y(t)\), thereby generating a smoother reset-triggering signal \(q(t)\). Moreover, the phase of the shaping filter, \(\varphi(\omega)=\angle C_s(j\omega)\), appears explicitly in the SIDF and HOSIDF expressions in~\eqref{eq: Hn}--\eqref{eq: HOSIDFs all functions}. This indicates that an appropriate choice of \(C_s\) can reduce the higher-order harmonic content generated by the reset element and, consequently, improve the reset-instant reliability metrics introduced in Section~\ref{sec: Metrics}.

In this section, we introduce a modified GFORE equipped with a first-order shaping filter. The shaping filter is selected such that the low-frequency contribution of the higher-order harmonics is further attenuated. In particular, the proposed choice increases the low-frequency slope of the nonzero HOSIDFs from \(+40\) dB/dec to \(+80\) dB/dec.

\begin{proposition}
\label{Pro: Shaping Filter}
Consider the reset element \(\mathcal{R}\) in~\eqref{eq: reset state space} as a first-order reset element with
\[
A_r=-\omega_r,\qquad B_r=1,\qquad C_r=\omega_r,\qquad D_r\in\mathbb{R},
\]
where \(\omega_r\in\mathbb{R}_{>0}\), and let \(A_\rho=\gamma\in(-1,1)\). Consider the first-order shaping filter
\begin{equation}
\label{eq: shaping Filter}
C_s:
\begin{cases}
\dot{x}_{\mathfrak{q}}(t)
=
-\dfrac{\omega_1\omega_r}{\omega_1+\omega_r}x_{\mathfrak{q}}(t)
+
\dfrac{\omega_1\omega_r}{\omega_1+\omega_r}y(t),
\\[8pt]
q(t)
=
\dfrac{\omega_1}{\omega_1+\omega_r}x_{\mathfrak{q}}(t)
+
\dfrac{\omega_r}{\omega_1+\omega_r}y(t),
\end{cases}
\end{equation}
with \(\omega_1\in\mathbb{R}_{>0}\). Then, for the nonzero higher-order harmonics, i.e., for odd \(n\geq 3\), the HOSIDFs satisfy
\begin{equation}
|H_n(\omega)|=\mathcal{O}(\omega^4),
\qquad \omega\to 0.
\end{equation}
In contrast, without the shaping filter, i.e., for \(C_s(s)=1\), the corresponding low-frequency behavior is
\begin{equation}
|H_n(\omega)|=\mathcal{O}(\omega^2),
\qquad \omega\to 0.
\end{equation}
Hence, the proposed shaping filter increases the low-frequency attenuation slope of the nonzero HOSIDFs from \(+40\) dB/dec to \(+80\) dB/dec.
\end{proposition}
\textbf{Proof:} See Appendix~\ref{app: proof shaping filter}.

The improved low-frequency attenuation of the nonzero HOSIDFs reduces the higher-order harmonic content generated by the reset element in the low-frequency range. This is particularly relevant for reset-instant reliability, since violations of the two-reset condition, reflected by \(N_r>2\) and large values of \(\sigma_d(\omega)\), may occur when higher-order harmonic contributions become dominant. Therefore, the proposed shaping filter helps reduce the tendency of the reset-triggering signal to generate additional zero crossings at low frequencies.

The next section examines this effect in a closed-loop setting. In particular, the proposed modified GFORE is implemented in the considered reset control architecture, and the resulting reset-triggering behavior is evaluated using the reliability metrics introduced earlier. This provides a bridge between the analytical low-frequency result derived in this section and the subsequent case study and experimental validation.

\section{Case Study}
\label{sec: case study}

The case study is carried out on an industrial wire-bonder motion platform. Wire bonders are widely used in semiconductor manufacturing to establish electrical connections between an integrated circuit and its package terminals. The isolated motion stage considered in this work is shown in Fig.~\ref{fig: Wire-bonder}. The platform provides three translational degrees of freedom along the X-, Y-, and Z-directions, each driven by a dedicated actuator.

In this study, the analysis is focused on the X-stage dynamics. To assess the degree of coupling between the motion axes, Fig.~\ref{Fig: FRF AB383} shows the measured frequency response functions (FRFs) from the actuator forces \(F_{\mathrm{x}}\), \(F_{\mathrm{y}}\) and \(F_{\mathrm{z}}\) to the measured X-direction displacement \(D_{\mathrm{x}}\). The cross-axis responses from \(F_{\mathrm{y}}\) to \(D_{\mathrm{x}}\) and from \(F_{\mathrm{z}}\) to \(D_{\mathrm{x}}\) are considerably smaller than the direct response from \(F_{\mathrm{x}}\) to \(D_{\mathrm{x}}\). In particular, their magnitudes remain approximately \(40~\mathrm{dB}\) below the direct X-axis response (around and below the targeted bandwidth frequency), corresponding to a difference of about two orders of magnitude. This indicates that the coupling from the Y- and Z-actuation channels to the X-displacement is negligible for the purpose of the present analysis. Therefore, the dynamics along the X-, Y-, and Z-axes can be treated as approximately decoupled SISO systems.

For confidentiality, the frequency axis in Fig.~\ref{Fig: FRF AB383}, as well as in all subsequent frequency-domain results, is normalized with respect to the Nyquist frequency. Specifically, frequencies in Hz are normalized by $F_s/2$, where $F_s$ denotes the sampling frequency, while angular frequencies are normalized by the Nyquist angular frequency $2\pi F_s/2$. To provide a practical indication of the physical frequency range, the sampling frequency of the experimental setup is on the order of $10~\mathrm{kHz}$.

\begin{figure}[!t]
\centering
\includegraphics[width=0.7\columnwidth]{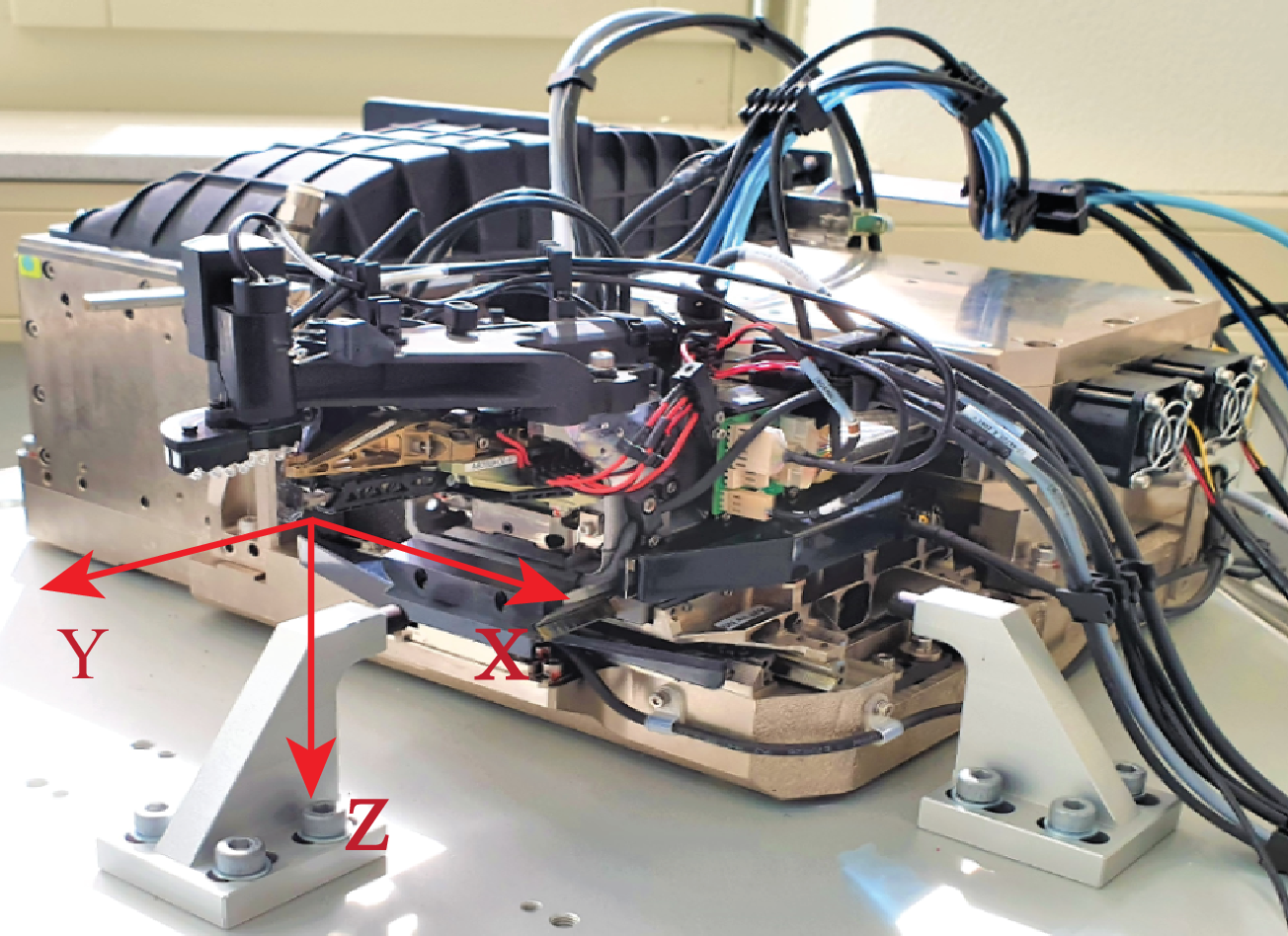}
\caption{Isolated XYZ-motion platform of the wire bonder.}
\label{fig: Wire-bonder}
\end{figure}

\begin{figure}[!t]
\centering
\includegraphics[width=\columnwidth]{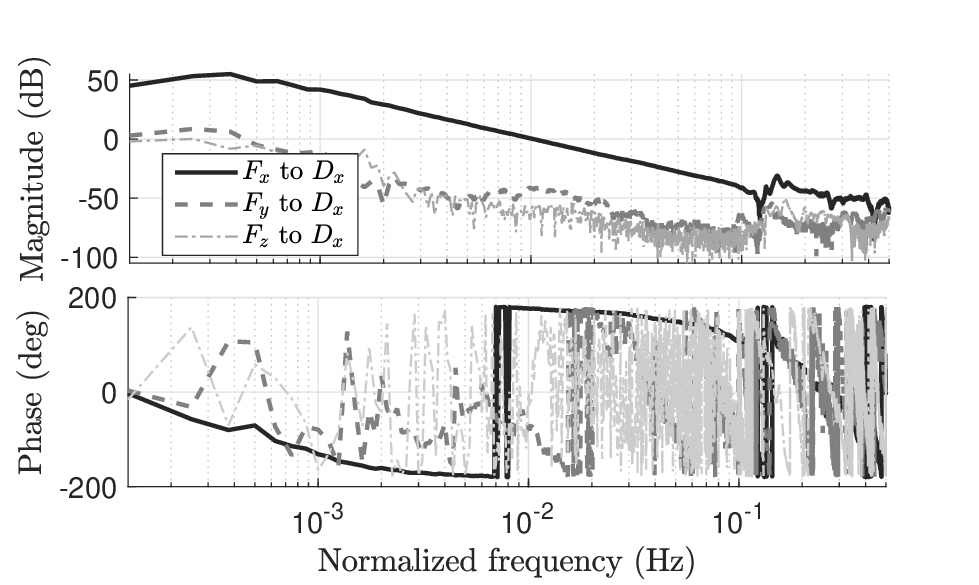}
\caption{Measured FRFs from the actuator forces \(F_{\mathrm{x}}\), \(F_{\mathrm{y}}\) and \(F_{\mathrm{z}}\) to the X-direction displacement \(D_{\mathrm{x}}\) of the wire-bonder motion stage.}
\label{Fig: FRF AB383}
\end{figure}

\subsection{Linear and Reset-Based Controller Design}
\label{subsec: controller_design}

In this section, a linear controller is first designed for the introduced motion stage. Subsequently, it is shown how reset control can be used to further improve the closed-loop performance, and how the proposed reliability metrics can support the reset-controller design process.

Let \(C_{\mathrm{L}}(s)\) denote the baseline LTI controller. The controller consists of a PID-type structure combined with a notch filter to attenuate the high-frequency resonance of the plant, and is given by
\begin{equation}
\label{eq: CL_controller}
C_{\mathrm{L}}(s)
=
k_p
\left(1+\frac{\omega_i}{s}\right)
\left(
\frac{1+\frac{\omega_d}{s}}
{1+\frac{\omega_t}{s}}
\right)
\left(
\frac{
1+\frac{s}{Q_1\omega_n}+\frac{s^2}{\omega_n^2}
}{
1+\frac{s}{Q_2\omega_n}+\frac{s^2}{\omega_n^2}
}
\right),
\end{equation}
where \(k_p\in\mathbb{R}_{>0}\) is the proportional gain, \(\omega_i\in\mathbb{R}_{>0}\) is the integral corner frequency, \(\omega_d\in\mathbb{R}_{>0}\) and \(\omega_t\in\mathbb{R}_{>0}\) define the lead-filter zero and pole frequencies, respectively, and \(\omega_n\in\mathbb{R}_{>0}\) denotes the notch frequency. The parameters \(Q_1,Q_2\in\mathbb{R}_{>0}\) determine the damping characteristics of the notch filter. 

For the X-stage plant \(P\), the baseline LTI controller is designed subject to constraints on the sensitivity peak below and around the resonance region. Specifically,
\begin{equation}
\label{eq: constraints}
M_{\mathrm{s}}\leq 6~\mathrm{dB},
\qquad
M_{\mathrm{r}}\leq 2.5~\mathrm{dB},
\end{equation}
where
\begin{align*}
\label{eq: Ms_Mr_def}
&M_{\mathrm{s}}
:=
20\log_{10}
\left(
\max_{\omega<\omega_{\mathrm{res}}}
|S(j\omega)|
\right),\\
&M_{\mathrm{r}}
:=
20\log_{10}
\left(
\max_{\omega\geq\omega_{\mathrm{res}}}
|S(j\omega)|
\right).
\end{align*}
Here,
\(
S(j\omega)=1/(1+C_{\mathrm{L}}(j\omega)P(j\omega))
\)
denotes the sensitivity function, and \(\omega_{\mathrm{res}}\in\mathbb{R}_{>0}\) is the frequency immediately preceding the first resonance or anti-resonance of the plant. The parameters of \(C_{\mathrm{L}}\), selected to maximize the closed-loop bandwidth while satisfying constraints in~\eqref{eq: constraints}, are reported in Table~\ref{tab:controller_parameters}.

The achievable bandwidth of the baseline LTI controller is mainly limited by the available phase margin. This limitation is caused, in part, by the delay present in the system. To increase the bandwidth without substantially changing the gain characteristics, a reset-based filter is introduced to provide additional phase. The adopted structure follows the reset-filter architecture introduced in~\cite[Definition~1]{hosseini2025AddOnFilterDesign} and is implemented according to the closed-loop configuration in Fig.~\ref{Fig: Block diagram CL}.

To this end, a high-bandwidth LTI controller \(C_{\mathrm{L}}'(s)\), with the same structure as~\eqref{eq: CL_controller}, is first considered. This controller is selected to increase the bandwidth, but by itself it does not provide sufficient phase margin and therefore leads to \(M_{\mathrm{s}}>6~\mathrm{dB}\). The reset-based controller \(C_{\mathrm{R}}\) is then obtained by choosing
\begin{equation}
C_{\mathrm{pos}}(s)=C_{\mathrm{L}}'(s),
\qquad
C_{\mathrm{par}}(s)=0,
\qquad
C_{\mathrm{pre}}(s)=k_{\mathrm{c}}C_{\mathfrak{c}}(s),
\end{equation}
where
\begin{equation}
\label{eq: kc}
k_{\mathrm{c}}
=
\frac{\omega_f-\omega_l}{\omega_f},
\qquad
C_{\mathfrak{c}}(s)
=
\frac{1+s/\omega_l}{1+s/\omega_f},
\end{equation}
with \(\omega_l,\omega_f\in\mathbb{R}_{>0}\) and \(\omega_l<\omega_f\). The interval \([\omega_l,\omega_f]\) specifies the frequency range in which the reset-based filter is intended to provide additional phase.

The reset element \(\mathcal{R}\) is chosen as a GFORE element with \(n_r=1\), given by
\begin{equation}
\label{eq: Dr}
A_r=-\omega_r,
\qquad
B_r=1,
\qquad
C_r=\omega_r,
\qquad
D_r=
\frac{\omega_l}{\omega_f-\omega_l},
\end{equation}
where
\begin{equation}
\label{eq: corner frequency}
\omega_r
=
\frac{\omega_l}
{\sqrt{
1+
\left(
\dfrac{4(1-A_\rho)}
{\pi(1+A_\rho)}
\right)^2
}}.
\end{equation}
The resulting controller parameters for the baseline LTI controller and the reset-based controller are summarized in Table~\ref{tab:controller_parameters}.

\begin{table}[!t]
\centering
\caption{Controller parameters for the baseline LTI controller \(C_{\mathrm{L}}\) and the reset-based controller \(C_{\mathrm{R}}\). Frequencies are reported in normalized angular-frequency units (rad/s).}
\label{tab:controller_parameters}
\renewcommand{\arraystretch}{1.15}
\begin{tabular}{c|c|c}
\hline
\textbf{Parameter} & \(\boldsymbol{C_{\mathrm{L}}}\) & \(\boldsymbol{C_{\mathrm{R}}}\) \\
\hline
\(k_p\)        & $1.129$ & $1.99$ \\
\(\omega_i\)  & $2.480 \times 10^{-2}$ & $3.272 \times 10^{-2}$ \\
\(\omega_d\)  & $3.927 \times 10^{-2}$ & $5.236 \times 10^{-2}$ \\
\(\omega_t\)  & $3.534 \times 10^{-1}$ & $4.712 \times 10^{-1}$ \\
\(\omega_n\)  & $9.975 \times 10^{-1}$ & $9.975 \times 10^{-1}$ \\
\(Q_1\)        & $4$ & $4$ \\
\(Q_2\)        & $3.5$ & $3.5$ \\
\hline
\(\omega_l\)  & -- & $8.639 \times 10^{-2}$ \\
\(\omega_f\)  & -- & $3.685 \times 10^{-1}$ \\
\(\omega_r\)  & -- & $5.336 \times 10^{-2}$ \\
\(A_\rho\)    & -- & $0$ \\
\hline
\end{tabular}
\end{table}

We first consider the reset-based controller without a shaping filter, i.e.,
\(C_s(s)=1\). Fig.~\ref{fig:sensitivity_noCs} compares the sensitivity function
of the baseline LTI controller with the first-order closed-loop sensitivity of
the reset-based controller. For the reset-based controller, the closed-loop
HOSIDFs from the reference input \(r\) to the error output \(e\) are obtained
by setting \(w=r\) and \(z=e\) in~\eqref{eq: Sz1_shaping} and
\eqref{eq: Szn_shaping}. The corresponding first- and third-order components
are denoted by \(S_{e,1}(\omega)\) and \(S_{e,3}(\omega)\), respectively.

As shown in Fig.~\ref{fig:sensitivity_noCs}, the reset-based controller increases
the closed-loop bandwidth compared with the baseline LTI controller. This improves
the low-frequency performance, for example in terms of reference tracking and
disturbance rejection, while still satisfying the sensitivity constraints
in~\eqref{eq: constraints}. However, \(S_{e,1}\) only describes the first
harmonic component of the closed-loop response. The third-order component
\(S_{e,3}\) shows that, at some frequencies, the higher-order harmonic content
is not negligible compared with the first harmonic. This may reduce the
reliability of the predicted reset instants and, consequently, lead to actual
closed-loop performance that differs from the frequency-domain prediction.

To assess this effect, the reset-instant reliability metrics \(\sigma_t\) and
\(\sigma_d\) are evaluated for controller $C_{\mathrm{R}}$. Fig.~\ref{fig:metrics_noCs}(a)
shows \(\sigma_t\), indicating that, at certain frequencies, the actual zero
crossings of the reset-triggering signal deviate from the zero crossings assumed
in the first-harmonic prediction by up to approximately \(30\%\) of the nominal
inter-reset interval. Fig.~\ref{fig:metrics_noCs}(b) shows \(\sigma_d\). The
frequency regions where \(\sigma_d>100\) correspond to cases with multiple zero
crossings, i.e., \(N_r>2\). These results indicate that, although the reset-based
controller improves the predicted first-harmonic performance, the reset instants
are not sufficiently reliable over the full frequency range. Hence, the predicted
improvement may not be fully achieved in practice. This demonstrates the relevance
of the proposed metrics for evaluating the reliability of a designed reset
controller.

To reduce the undesired nonlinear contribution and improve the reliability of
the reset-triggering signal, the shaping-filter design proposed in
Proposition~\ref{Pro: Shaping Filter} is applied. Using the value of
\(\omega_r\) reported in Table~\ref{tab:controller_parameters}, the shaping
filter is selected to increase the low-frequency attenuation of the higher-order
harmonics. The resulting controller, denoted by \(C_{\mathrm{R,sf}}\), has the
same structure and parameters as \(C_{\mathrm{R}}\), except that the reset
element is equipped with the shaping filter \(C_s\) in
Proposition~\ref{Pro: Shaping Filter}. In this design,
\(\omega_r=0.0533~\mathrm{rad/s}\) and
\(\omega_1=0.062~\mathrm{rad/s}\).

The effect of the shaping filter can be observed in
Fig.~\ref{fig:sensitivity_noCs}. The first-order sensitivity
\(S_{e,1}(\omega)\) remains nearly unchanged for \(C_{\mathrm{R}}\) and
\(C_{\mathrm{R,sf}}\), indicating that the dominant closed-loop response is
preserved. In contrast, the third-order sensitivity is reduced for
\(C_{\mathrm{R,sf}}\), with the low-frequency attenuation behavior predicted by
Proposition~\ref{Pro: Shaping Filter}. This reduction in the higher-order
harmonic contribution improves the reliability of the reset-triggering signal.

To further assess this improvement, the metrics \(\sigma_t\) and \(\sigma_d\)
are evaluated for \(C_{\mathrm{R,sf}}\) and compared with the case without the
shaping filter in Fig.~\ref{fig:metrics_noCs}. The results show that
\(C_{\mathrm{R,sf}}\) reduces the reset-time deviation measured by
\(\sigma_t\). Moreover, the multiple-zero-crossing issue is eliminated over the
considered frequency range, as \(N_r=2\) for all studied frequencies. This is
also reflected in the improved values of \(\sigma_d\), which remain below the
violation threshold associated with \(N_r>2\).

These results demonstrate the practical relevance of the proposed reliability
metrics in the reset-controller design process. In particular, the metrics
provide frequency-domain indicators for identifying reset controllers that may
produce unreliable reset instants and, consequently, degraded time-domain
performance. Therefore, in the next section, the three controllers
\(C_{\mathrm{L}}\), \(C_{\mathrm{R}}\), and \(C_{\mathrm{R,sf}}\) are implemented
on the wire-bonder motion stage to validate the frequency-domain observations
in time-domain experiments.

\begin{figure}[!t]
\centering
\includegraphics[width=\columnwidth]{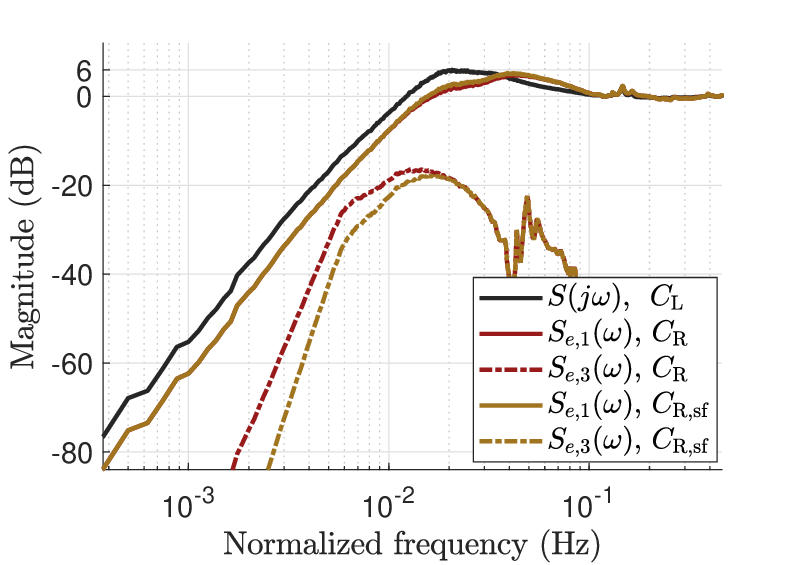}
\caption{Sensitivity comparison for the baseline LTI controller and the reset-based controller with and without the shaping filter. The first- and third-order closed-loop HOSIDFs from \(r\) to \(e\), denoted by \(S_{e,1}\) and \(S_{e,3}\), are also shown for the reset-based controllers.}
\label{fig:sensitivity_noCs}
\end{figure}

\begin{figure}[!t]
\centering
\subfloat[]{%
\includegraphics[width=1\columnwidth]{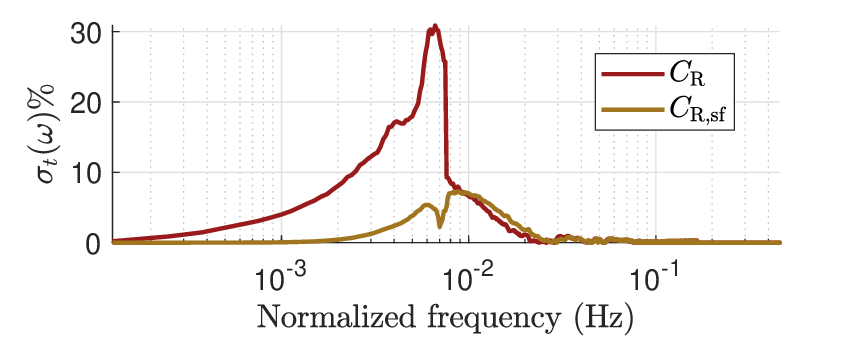}%
\label{fig:sigma_t_noCs}}
\hfil
\subfloat[]{%
\includegraphics[width=0.99\columnwidth]{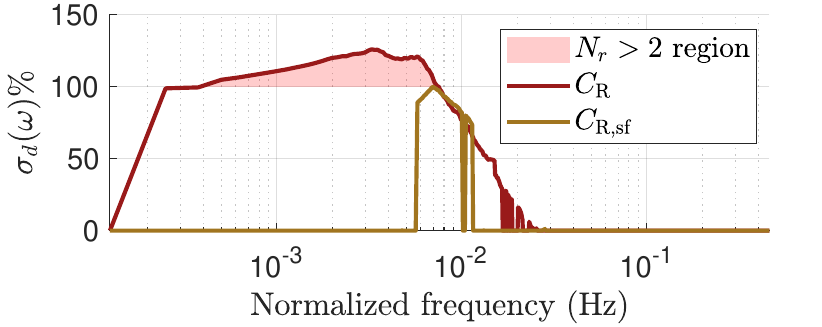}%
\label{fig:sigma_d_noCs}}
\caption{Reliability metrics for the reset-based controller with and without the shaping filter:  
(a) reset-time deviation metric \(\sigma_t\),  
(b) zero-crossing reliability metric \(\sigma_d\). Regions with \(\sigma_d>100\) indicate multiple zero crossings, i.e., \(N_r>2\).}
\label{fig:metrics_noCs}
\end{figure}

\section{Experimental Validation}\label{sec: Experimental Validation}

The frequency-domain observations from the case study are validated experimentally on the X-stage of the wire-bonder motion platform. The three controllers \(C_{\mathrm{L}}\), \(C_{\mathrm{R}}\), and \(C_{\mathrm{R,sf}}\) are implemented and tested under sinusoidal reference excitations of the form $r(t)=\hat{r}\sin(2\pi f_{r,i}t)$, with
\begin{equation}
f_{r,1}=0.003125\, \mathrm{Hz},
\qquad
f_{r,2}=0.00625\, \mathrm{Hz}.
\end{equation}

The excitation frequencies are selected from the range where the reliability metrics clearly distinguish between reset-based controllers with and without the shaping filter. The reference amplitude \(\hat{r}\) is not reported explicitly, since the time-domain signals are normalized for confidentiality. Nevertheless, to provide physical context, the reference amplitudes are on the order of a few hundred micrometers, while the measured tracking errors are in the sub-micrometer range, which is consistent with the accuracy requirements of the considered precision motion stage.

Fig.~\ref{fig:exp_error} shows the measured tracking error \(e(t)\) for the three controllers, and the corresponding RMS values are reported in Table~\ref{tab:rms_error_exp}. For \(f_{r,1}\), the reset-based controllers \(C_{\mathrm{R}}\) and \(C_{\mathrm{R,sf}}\) reduce the RMS error by approximately \(49.4\%\) and \(50.0\%\), respectively, compared with the baseline LTI controller. For \(f_{r,2}\), the corresponding reductions are approximately \(36.5\%\) and \(41.0\%\). Hence, both reset-based controllers improve the measured tracking performance. It should be emphasized, however, that the objective here is not to optimize the tracking-performance difference between \(C_{\mathrm{R}}\) and \(C_{\mathrm{R,sf}}\), but rather to compare their reset-triggering reliability.

Thus, Fig.~\ref{fig:exp_q} compares the measured reset-triggering signal \(q(t)\) for
\(C_{\mathrm{R}}\) and \(C_{\mathrm{R,sf}}\). For both excitation frequencies,
the shaping filter leads to a visibly smoother reset-triggering signal. This
behavior is consistent with the frequency-domain observations in the design
section and with the reliability metrics shown in Fig.~\ref{fig:metrics_noCs},
where the inclusion of the shaping filter reduced the higher-order harmonic
contribution and eliminated the occurrence of multiple zero crossings over the
considered frequency range. Consequently, the reset events generated by
\(C_{\mathrm{R,sf}}\) are more consistent with the first-harmonic reset
assumption used in the closed-loop HOSIDF analysis. These experimental results
support the main conclusion of the case study: the proposed reliability metrics
can identify reset designs with potentially unreliable reset-triggering behavior.

All controllers are implemented digitally. The LTI components are discretized using the Tustin approximation, which is adopted to preserve the phase characteristics of the continuous-time dynamics over the relevant frequency range~\cite{aastrom2013computerTustin}. Accordingly, the frequency-domain results reported for the LTI parts are evaluated using their discrete-time frequency responses. The reset element is discretized using the same approach; details of the resulting discrete-time GFORE implementation can be found in~\cite[Section~V.B]{hosseini2025AddOnFilterDesign}. The closed-loop stability of the implemented reset controllers was verified using the frequency-domain stability assessment in~\cite{dastjerdi2023frequency}. All experiments were performed within the verified stable operating regime. A complete stability synthesis for reset control systems is beyond the scope of this paper.

\begin{table}[!t]
\centering
\caption{Measured RMS tracking error for sinusoidal reference-tracking experiments.}
\label{tab:rms_error_exp}
\renewcommand{\arraystretch}{1.15}
\resizebox{\columnwidth}{!}{%
\begin{tabular}{c|c|c|c}
\hline
\textbf{Frequency (Hz)} &
\(\boldsymbol{C_{\mathrm{L}}}\) &
\(\boldsymbol{C_{\mathrm{R}}}\) &
\(\boldsymbol{C_{\mathrm{R,sf}}}\) \\
\hline
\(f_{r,1}=3.125\times10^{-3}\) & \(1.621\times10^{-6}\) & \(8.204\times10^{-7}\) & \(8.105\times10^{-7}\) \\
\(f_{r,2}=6.25\times10^{-3}\)  & \(8.468\times10^{-6}\) & \(5.378\times10^{-6}\) & \(4.993\times10^{-6}\) \\
\hline
\end{tabular}%
}
\end{table}

\begin{figure}[!t]

\centering

\subfloat[]{%

\includegraphics[width=0.95\columnwidth]{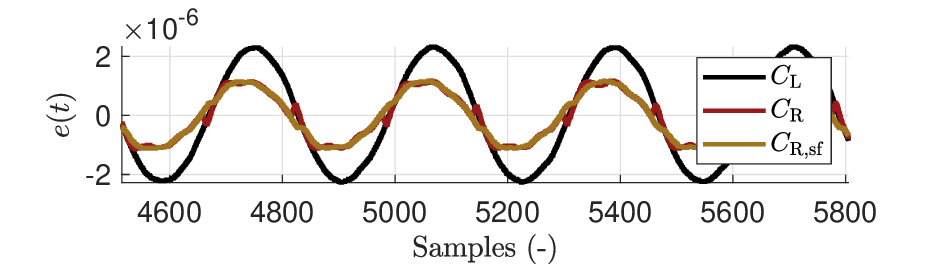}%

\label{fig:error_fr1}}

\vspace{1mm}

\subfloat[]{%

\includegraphics[width=0.95\columnwidth]{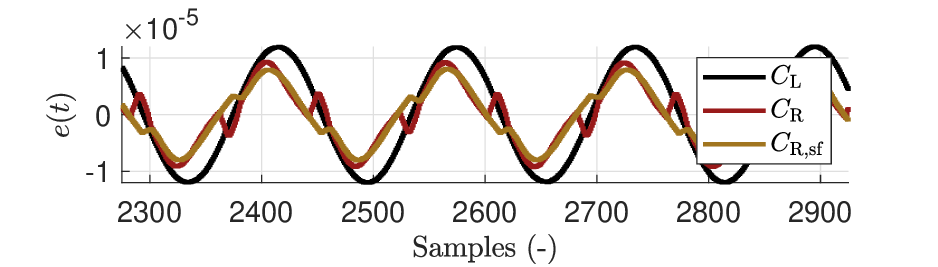}%

\label{fig:error_fr2}}

\caption{Measured tracking error \(e(t)\) for \(C_{\mathrm{L}}\), \(C_{\mathrm{R}}\), and \(C_{\mathrm{R,sf}}\): (a) \(f_{r,1}=0.003125\, \mathrm{Hz}\), (b) \(f_{r,2}=0.00625\, \mathrm{Hz}\).}

\label{fig:exp_error}

\end{figure}

\begin{figure}[!t]

\centering

\subfloat[]{%

\includegraphics[width=0.95\columnwidth]{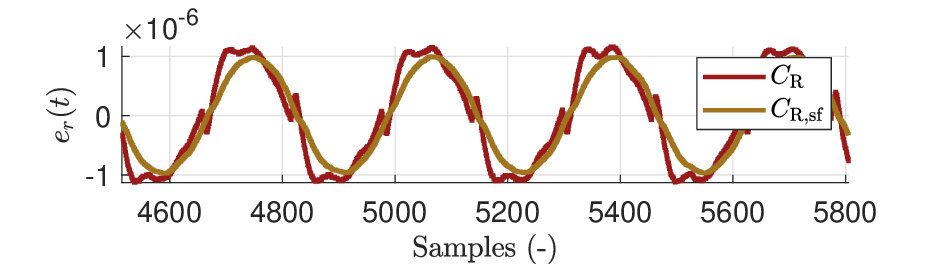}%

\label{fig:q_fr1}}

\vspace{1mm}

\subfloat[]{%

\includegraphics[width=0.95\columnwidth]{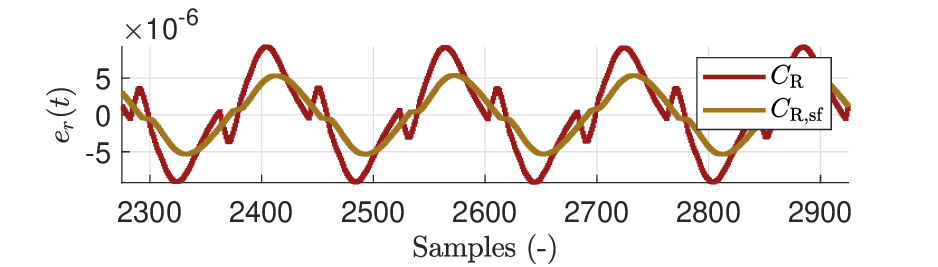}%

\label{fig:q_fr2}}

\caption{Measured reset-triggering signal \(q(t)\) for \(C_{\mathrm{R}}\) and \(C_{\mathrm{R,sf}}\): (a) \(f_{r,1}=0.003125\, \mathrm{Hz}\), (b) \(f_{r,2}=0.00625\, \mathrm{Hz}\).}

\label{fig:exp_q}

\end{figure}
\section{Conclusion}\label{sec: Conclusion}

This paper addressed the reliability of reset instants in frequency-domain reset-control design. Although HOSIDFs provide information about the higher-order harmonic content generated by reset elements, their impact on the actual reset-triggering behavior is not immediately evident from the HOSIDF magnitudes alone. This is important because the performance advantage of reset control originates from its nonlinear reset action; therefore, the timing and regularity of the reset instants must also be considered when assessing a design. The proposed metrics make this connection explicit by translating the harmonic content of the reset-triggering signal into indicators of reset-time deviation and proximity to additional zero crossings. As shown in the case study, a reset controller may appear favorable from the first-order closed-loop response while still exhibiting unreliable reset-triggering behavior. The metrics, therefore, provide an additional design layer beyond conventional sensitivity-based assessment.

The experimental results on the motion stage support this interpretation. Both reset-based controllers reduced tracking error relative to the baseline LTI controller, while the shaped reset controller produced a smoother reset-triggering signal, consistent with the reliability analysis. Hence, the proposed framework helps evaluate not only whether reset control improves performance, but also whether the reset actions responsible for that improvement occur in a reliable and predictable manner. Future work will investigate higher-order shaping-filter designs to further reduce undesired nonlinearities and study how different filter placements affect the zero-crossing behavior of the reset-triggering signal.




\appendix   
\section{Proof of Theorem~\ref{thm:closed_loop_shaping}}\label{sec: App A}

The expressions for \(y(t)\) and \(z(t)\) follow directly from the closed-loop
frequency-domain relations derived in~\cite{LukeLure}. In particular, the
additional output row \([G_{wq}(s)\;\;G_{uq}(s)]\) used to define \(q(t)\) does not alter the interconnection
between \(w\), \(u\), \(y\), and \(z\). Therefore, the closed-loop relations from
which \(Y_n(\omega)\) and \(Z_n(\omega)\) are obtained remain unchanged.
The only difference is that the input-output behavior of the reset element is now
described by the SIDF/HOSIDFs \(H_n(\omega)\) of the reset element with shaping
filter, as defined in~\eqref{eq: Hn}. Substituting these \(H_n(\omega)\) into the
closed-loop expressions of~\cite{LukeLure} gives
\eqref{eq: y_periodic_shaping}--\eqref{eq: Syn_shaping} and
\eqref{eq: z_periodic_shaping}--\eqref{eq: Szn_shaping}.

It remains to derive the harmonic response of the reset-triggering signal
\(q(t)\). Let \(W(j\omega)=\hat{w}e^{j\phi_w}\). From the first harmonic of the
extended LTI relation in~\eqref{eq: G_IO_shaping},
\begin{equation}
Q_1(\omega)
=
G_{wq}(j\omega)W(j\omega)
+
G_{uq}(j\omega)U_1(\omega).
\end{equation}
Using \(U_1(\omega)=H_1(\omega)Y_1(\omega)\) and
\(Y_1(\omega)=S_{y,1}(\omega)W(j\omega)\), it follows that
\begin{equation}
Q_1(\omega)=S_{q,1}(\omega)W(j\omega),
\end{equation}
with \(S_{q,1}\) as in~\eqref{eq: Sq1_shaping}.

For \(n\neq 1\), the external input has no harmonic component at \(n\omega\), and
therefore
\begin{equation}
Q_n(\omega)=G_{uq}(nj\omega)U_n(\omega).
\end{equation}
From the harmonic-separator approximation and the result for \(Y_n(\omega)\),
\begin{equation}
U_n(\omega)
=
\frac{
H_n(\omega)S_{y,1}(\omega)
e^{j(n-1)\angle S_{y,1}(\omega)}
}
{1-G_{uy}(nj\omega)R_{\mathrm{bl}}(nj\omega)}
W(j\omega)e^{j(n-1)\angle W(j\omega)}.
\end{equation}
Substituting this expression into \(Q_n(\omega)=G_{uq}(nj\omega)U_n(\omega)\)
gives
\begin{align}
Q_n(\omega)
&=
G_{uq}(nj\omega)
\frac{
H_n(\omega)S_{y,1}(\omega)
e^{j(n-1)\angle S_{y,1}(\omega)}
}
{1-G_{uy}(nj\omega)R_{\mathrm{bl}}(nj\omega)}\dots\nonumber\\
&\qquad \qquad \,\, W(j\omega)e^{j(n-1)\angle W(j\omega)}
\nonumber\\
&=
S_{q,n}(\omega)W(j\omega)e^{j(n-1)\angle W(j\omega)},
\qquad n\neq 1,
\end{align}
where
\begin{equation}
S_{q,n}(\omega)
=
\frac{
G_{uq}(nj\omega)H_n(\omega)S_{y,1}(\omega)
e^{j(n-1)\angle S_{y,1}(\omega)}
}
{1-G_{uy}(nj\omega)R_{\mathrm{bl}}(nj\omega)}.
\end{equation}
For the first harmonic, one similarly obtains
\begin{equation}
Q_1(\omega)=S_{q,1}(\omega)W(j\omega),
\end{equation}
with
\begin{equation}
S_{q,1}(\omega)
=
G_{wq}(j\omega)
+
G_{uq}(j\omega)H_1(\omega)S_{y,1}(\omega).
\end{equation}
This concludes the proof.

\section{Interpretation of the Critical-Point Conditions}
\label{app: stationary_point_conditions}

This appendix clarifies the conditions used in the definitions of
$\mathcal{E}_{2}(\omega)$ and $\mathcal{E}_{>2,k}(\omega)$. These sets are
introduced to identify the stationary points of the reset-triggering signal
$q(t,\omega)$ that are relevant to the formation or severity of additional
zero crossings. In particular, the creation or removal of a pair of zero
crossings occurs when a local extremum approaches and eventually touches the
zero level. Therefore, the proposed conditions are formulated in terms of
$\dot q(t,\omega)$, $\ddot q(t,\omega)$, and the sign of $q(t,\omega)$.

First, the condition
\begin{equation}
    \dot q(t,\omega)=0
\end{equation}
identifies the stationary points of $q(t,\omega)$. When $N_r=2$, a stationary
point is relevant to the possible creation of additional zero crossings if it
is either a local minimum above zero or a local maximum below zero. These two
cases are compactly characterized by
\begin{equation}
    q(t,\omega)\ddot q(t,\omega)>0.
\end{equation}
Indeed, if $q(t,\omega)>0$ and $\ddot q(t,\omega)>0$, the stationary point is
a positive local minimum. Conversely, if $q(t,\omega)<0$ and
$\ddot q(t,\omega)<0$, it is a negative local maximum. In both cases, the
extremum lies on the side of the zero level toward which an additional pair of
zero crossings can emerge. As $|q(t,\omega)|$ decreases, the extremum
approaches zero, and at $q(t,\omega)=0$ the signal reaches the boundary at
which an additional pair of zero crossings can be created. This motivates the
definition of $\mathcal{E}_{2}(\omega)$ and the use of the smallest
$|q(t,\omega)|$ over these critical stationary points as a measure of
proximity to violation of the two-reset condition.

When $N_r>2$, additional zero crossings have already appeared. The relevant
stationary point is then the extremum contained within the additional
zero-crossing excursion. Such a point is either a local maximum above zero or
a local minimum below zero, which is characterized by
\begin{equation}
\label{eq: B3}
    q(t,\omega)\ddot q(t,\omega)<0.
\end{equation}
Specifically, $q(t,\omega)>0$ together with $\ddot q(t,\omega)<0$
corresponds to a positive local maximum, whereas $q(t,\omega)<0$ together
with $\ddot q(t,\omega)>0$ corresponds to a negative local minimum. The
magnitude $|q(t,\omega)|$ at this extremum quantifies the depth of the
additional zero-crossing excursion and therefore the severity of the
violation of the two-reset condition.

The boundary case in which the relevant stationary point lies exactly on the
zero level should also be included. Therefore, the condition in
\eqref{eq: B3} is extended in \eqref{eq: Egt2_set} as
\begin{equation*}
    q(t,\omega)\ddot q(t,\omega)<0
    \;\vee\;
    q(t,\omega)=0.
\end{equation*}
The second condition accounts for the limiting case in which the local
extremum coincides with the zero level.

Finally, the condition
\begin{equation}
    \dot q_1(t,\omega)\ddot q(t,\omega)>0
\end{equation}
selects the stationary point associated with the additional zero-crossing
excursion around the nominal reset instant $t_k$. If the first harmonic
$q_1(t,\omega)$ has a positive slope around the nominal reset instant, this
condition requires $\ddot q(t,\omega)>0$ and therefore selects a local
minimum of $q(t,\omega)$. Conversely, if $q_1(t,\omega)$ has a negative
slope, the condition requires $\ddot q(t,\omega)<0$ and selects a local
maximum. Hence, this condition selects the extremum whose orientation is
consistent with the deviation from the nominal first-harmonic reset
transition. Together with the restriction
$t\in\mathcal{I}_k(\omega)$, it excludes unrelated stationary points elsewhere
in the period and retains only those associated with the additional zero
crossings around $t_k$.
\section{Proof of Proposition \ref{Pro: Shaping Filter}}
\label{app: proof shaping filter}

From the state-space realization in~\eqref{eq: shaping Filter}, the transfer function from \(y(t)\) to \(q(t)\) is
\begin{equation}
C_s(s)
=
\frac{1+s/\omega_1}
{1+s(\omega_1+\omega_r)/(\omega_1\omega_r)}.
\label{eq: Cs_transfer_selected}
\end{equation}
Therefore,
\begin{equation}
\varphi(\omega)
=
\angle C_s(j\omega)
=
\arctan\!\left(\frac{\omega}{\omega_1}\right)
-
\arctan\!\left(\frac{\omega(\omega_1+\omega_r)}{\omega_1\omega_r}\right).
\label{eq: Cs_phase_selected}
\end{equation}
Let \(\zeta(\omega):=\tan(\varphi(\omega))\). Using the tangent subtraction identity gives
\begin{equation}
\zeta(\omega)
=
-\frac{\omega\omega_1^2}
{\omega_1^2\omega_r+(\omega_1+\omega_r)\omega^2}.
\label{eq: zeta_selected}
\end{equation}

For the scalar first-order reset element, the term in~\eqref{eq: HOSIDFs all functions} that determines the low-frequency order of the HOSIDFs is
\begin{equation}
\omega\cos\varphi(\omega)-A_r\sin\varphi(\omega)
=
\omega\cos\varphi(\omega)+\omega_r\sin\varphi(\omega).
\end{equation}
Since \(\cos(\arctan \zeta)=1/\sqrt{1+\zeta^2}\) and
\(\sin(\arctan \zeta)=\zeta/\sqrt{1+\zeta^2}\), it follows that
\begin{align}
\omega\cos\varphi+\omega_r\sin\varphi
&=
\frac{\omega+\omega_r\zeta(\omega)}
{\sqrt{1+\zeta^2(\omega)}} \nonumber\\
&=
\frac{(\omega_1+\omega_r)\omega^3}
{\omega_1^2\omega_r+(\omega_1+\omega_r)\omega^2}
\frac{1}{\sqrt{1+\zeta^2(\omega)}}.
\label{eq: cancellation_term}
\end{align}
Hence,
\begin{equation}
\omega\cos\varphi+\omega_r\sin\varphi
=
\mathcal{O}(\omega^3),
\qquad \omega\to 0.
\label{eq: cubic_order}
\end{equation}

Now consider the nonzero higher-order harmonics, i.e., odd \(n\geq 3\). From~\eqref{eq: Hn},
\begin{equation}
H_n(\omega)
=
C_r(A_r-jn\omega)^{-1}B_r\Theta_\varphi(\omega).
\end{equation}
For \(A_r=-\omega_r\), \(B_r=1\), and \(C_r=\omega_r\),
\begin{equation}
\left|
C_r(A_r-jn\omega)^{-1}B_r
\right|
=
\frac{\omega_r}{\sqrt{\omega_r^2+n^2\omega^2}}
=
1+\mathcal{O}(\omega^2).
\label{eq: base_linear_asymptotic}
\end{equation}
Furthermore, as \(\omega\to0\),
\begin{equation}
\Omega(\omega)=1-\gamma+\mathcal{O}\!\left(e^{-\pi\omega_r/\omega}\right),
\qquad
\Lambda(\omega)=\omega_r^2+\mathcal{O}(\omega^2).
\label{eq: Omega_Lambda_asymptotic}
\end{equation}
Using~\eqref{eq: HOSIDFs all functions} together with~\eqref{eq: cubic_order} gives
\begin{equation}
|\Theta_\varphi(\omega)|
=
\frac{2\omega}{\pi}
\frac{(1-\gamma)}{\omega_r^2}
\mathcal{O}(\omega^3)
=
\mathcal{O}(\omega^4).
\end{equation}
Therefore,
\begin{equation}
|H_n(\omega)|=\mathcal{O}(\omega^4),
\qquad n=3,5,\ldots,
\end{equation}
which corresponds to a low-frequency slope of \(+80\) dB/dec.

For comparison, without the shaping filter, \(C_s(s)=1\) and thus \(\varphi(\omega)=0\). In that case,
\begin{equation}
\omega\cos\varphi+\omega_r\sin\varphi=\omega.
\end{equation}
Repeating the same asymptotic argument gives
\begin{equation}
|\Theta_0(\omega)|
=
\frac{2(1-\gamma)}{\pi\omega_r^2}\omega^2
+
\mathcal{O}(\omega^4),
\end{equation}
and consequently
\begin{equation}
|H_n(\omega)|=\mathcal{O}(\omega^2),
\qquad n=3,5,\ldots,
\end{equation}
which corresponds to a low-frequency slope of \(+40\) dB/dec. This proves the result.





\bibliographystyle{asmejour}   

\bibliography{Main} 



\end{document}